\documentclass[draftclsnofoot,12pt,onecolumn]{IEEEtran}        
\usepackage{amsmath, amsthm, amssymb}
\usepackage{epsfig}
\usepackage{latexsym}
\usepackage{graphicx}
\usepackage{cite}
\usepackage{epsfig}
\usepackage{subfigure}
\usepackage{color}

\newtheorem{theorem}{Theorem}
\newtheorem{lemma}{Lemma}
\newtheorem{proposition}{Proposition}

\newtheorem{remark}{Remark}

\title{Exact Error Analysis and Energy-Efficiency Optimization of Regenerative   Relay Systems with  Spatial Correlation}

\author{Mulugeta~K.~Fikadu,~Paschalis~C.~Sofotasios,~Qimei~Cui,~Mikko~Valkama, and~George~K.~Karagiannidis

\thanks{M. K. Fikadu and M. Valkama are with the Department of Electronics and Communications Engineering, Tampere University of Technology, FI-33101 Tampere, Finland  \,(e-mail: $\rm \left\lbrace mulugeta.fikadu; mikko.e.valkama \right\rbrace@\rm tut.fi$) }

\thanks{P. C. Sofotasios was with the School of Electronic and Electrical Engineering, University of Leeds, LS2 9JT Leeds, UK.  He is now with the Department of Electronics and Communications Engineering, Tampere University of Technology, 33101 Tampere, Finland and with the Department of Electrical and Computer Engineering, Aristotle University of Thessaloniki, 54124 Thessaloniki, Greece  \, (e-mail: {$\rm p.sofotasios@ieee.org$)}  }

\thanks{Q. Cui is with the Wireless Technology Innovation Institute, Beijing University of Posts and Telecommunications, 100876 Beijing,  China \, (e-mail: $ \rm cuiqimei@bupt.edu.cn$)}

\thanks{G. K. Karagiannidis is with the Department of Electrical and Computer Engineering, Khalifa University, PO Box 127788
Abu Dhabi, UAE and with the Department of Electrical and Computer Engineering, Aristotle University of Thessaloniki, 54124 Thessaloniki, Greece \, (e-mail: $ \rm geokarag@ieee.org$)}
}
\begin{document}

\maketitle

\begin{abstract}
Energy efficiency and its  optimization constitute critical tasks in the design of low-power wireless networks. The present work  is devoted to the error rate analysis and energy-efficiency optimization of regenerative cooperative networks  in the presence of multipath fading   under spatial correlation. To this end,  exact and asymptotic analytic expressions are firstly derived for the  symbol-error-rate of  $M{-}$ary quadrature amplitude  and $M{-}$ary phase shift keying modulations assuming a dual-hop decode-and-forward  relay system, spatially correlated Nakagami${-}m$ multipath fading and maximum ratio combining. The derived  expressions are subsequently employed in   quantifying  the energy consumption of the considered system, incorporating both transmit energy and the energy consumed by the transceiver circuits, as well as in deriving the optimal power allocation formulation for minimizing  energy consumption under certain quality-of-service requirements.  A relatively harsh path-loss model, that also accounts for realistic device-to-device communications, is adopted in numerical evaluations  and various useful insights are provided for the design of future low-energy wireless networks deployments. Indicatively, it is shown that depending on the degree of spatial correlation, severity of fading, transmission distance, relay location and power allocation strategy, target  performance can be achieved with large overall energy reduction  compared to direct transmission reference. 
 \end{abstract}
 

\section{Introduction}

Emerging communication systems are expected to provide high-speed data transmission, efficient wireless access, high quality of service (QoS) and reliable network coverage with reduced processing time and energy as well as widespread use  of smart phones and other intelligent mobile devices. However, the currently witnessed scarcity of the two core fundamental resources, power and bandwidth, constitutes a significant challenge to satisfy these demands while  it is known that wireless channel impairments such as multipath fading, shadowing and interference degrade   information  signals during wireless propagation. Furthermore, most  energy constrained devices, such as terminals of mobile cellular, ad-hoc and wireless sensor networks, are typically powered by small batteries where replacement is   rather  difficult and costly   \cite{I},\cite{GE}.  Therefore,  finding a robust strategy for energy efficient transmission and minimized energy consumption per successfully communicated information bit is   essential in   effective design and deployment of wireless systems. This accounts for  example for cases such as  low-energy sensor networks in ecological environment monitoring  as well as energy consumption in infrastructure devices of cellular systems. In addition,  it is in line with global  policies and strategies on low energy consumption and  awareness  on environmental issues which, among others,  has led to the rapid emerge of green communications \cite{SE},\cite{SEE}.

It has been shown that multi-antenna  systems constitute an effective method that can  enhance spectral efficiency. However, this typically comes at a cost of complex transceiver circuitry  and in massive systems  with high energy consumption requirements. Furthermore, it is not currently  feasible to  embody   large multi-antenna systems at  hand held terminals due to spatial restrictions. As a result,  cooperative communications have been proposed as an alternative solution that improves coverage as well as performance under fading effects and have attracted significant attention due to their  ability to overcome the limitations  of resource constrained wireless access networks, see e.g. \cite{C, S, J, Final_1, Final_2, Final_3, Final_4, Final_5, Final_6, Final_7, Final_8, Final_9, Final_10, Final_11} and the references therein. 
A distinct feature of  cooperative communications is that wireless agents share resources, instead of competing for them, which  ultimately enhances system performance.  In this context,  various resource allocation algorithms and   techniques have been proposed for improving the energy efficiency of resource constrained wireless networks. Specifically, the authors in \cite{Bahai} analyzed energy-efficient  direct transmission adopting higher-level modulation for short distances, where circuit power is more dominant than transmission power. It was also suggested that high   energy reduction can be achieved by optimizing the transmission time and the modulation parameters, particularly for short transmission distances. The authors in \cite{R} addressed the optimal power allocation and throughput  transmission strategy for minimizing the total energy consumption required to transmit a given number of bits. In \cite{Z}, minimization of two-hop transmission energy with joint relay selection and power control was proposed for two policies:  i) for minimizing the energy consumption per data packet; ii) for maximizing the network lifetime. In the same context, \cite{W, X:Goldsmith, RQ, QAM, QAM2} addressed the modulation optimization for minimizing the total energy consumption for $M{-}$ary quadrature amplitude modulation ($M{-}$QAM), whereas energy efficient cooperative communication in clustered sensor networks was investigated  in \cite{A}. Energy efficiency in cooperative  networks was also analyzed in \cite{GL,YZ,YX,KM,K,WW,EES, EE_3} by  optimizing energy consumption based on the involved relay decoding strategy, modulation parameters, number of relay nodes and their distance from  the source and the destination nodes.  Likewise, an energy-efficient  scheme was proposed in \cite {WE} by exploiting  the  wireless  broadcast  nature  and  the  node overhearing capability while an optimal energy efficient strategy based on the cooperative network parameters and transmission rate was reported in \cite {GG}. Finally,  the authors in \cite{IV} analyzed  realistic scenarios of  energy efficient infrastructure-to-vehicle communications.            

It is also widely  known that fading phenomena constitute a crucial factor of performance degradation in conventional and emerging wireless communication systems. Based on this, numerous investigations have addressed the effect of different types of fading conditions on the performance of cooperative communications \cite{ Costa_3, Trung, New_1, New_4, Shi, Add_3,Add_4, Sofotasios_1, Sofotasios_2, Sofotasios_3, Sofotasios_4, Sofotasios_5, Sofotasios_6, Sofotasios_7, Sofotasios_8, Sofotasios_9, Sofotasios_10, Sofotasios_11}. However,   the vast majority of the reported investigations assume  that the involved communication paths are statistically independent to each-other. Nevertheless, this assumption is rather  simplistic as in realistic cooperative  communication scenarios the wireless channels may be spatially correlated, which should be taken into account particularly for deployments relating to low-energy consumption requirements. Based on this, the authors in \cite{D:Lee} addressed the spatial correlation in relay communications over  fading channels whereas the   performance of a decode-and-forward (DF)  system with  $M{-}$PSK modulated signals  in  triple correlated branches over   Nakagami${-}m$ fading channels using selection combining  was investigated in \cite{SWR}. In the same context, the authors in \cite{KYY} analyzed the performance of a    multiple-input-multiple-output (MIMO)  DF system with orthogonal space time transmission over spatially correlated Nakagami${-}m$ fading channels for integer values of $m$.   The performance of   a two hop amplify and forward (AF)  relay network with  beamforming and spatial correlation for the case that the  source and destination are equipped with multiple antennas  while the relay is equipped with a single antenna was investigated in  \cite{RaymondH}. Likewise,  spatial correlation in the context of indoor office environments and multi-antenna AF relaying with keyhole effects  was  analyzed in \cite{indoor} and \cite{Theo_keyhole}, respectively whereas the effects of spatial correlation on the performance of similar relaying systems were analyzed in \cite{YA, HKR, KY, MPI}. However, to the best of the authors knowledge, a comprehensive exact   and asymptotic error rate  analysis   for   regenerative systems   over spatially correlated channels using maximum-ratio-combing (MRC) as well as a detailed energy-efficiency analysis and optimization, have not been reported in the open technical literature.

Motivated by the above, the aim of this work is twofold: we, firstly, derive  exact analytic expressions for the SER of a two-hop DF relay system over spatially correlated Nakagami${-}m$ fading channels for both $M{-}$QAM and $M{-}$PSK constellations along with simple  and accurate asymptotic expressions for high signal-to-noise ratio (SNR) values. Secondly, we provide a comprehensive analysis of energy-efficiency and the corresponding optimization in terms of power allocation between cooperating devices. This is realized by minimizing the average total energy consumption of the  DF relay network over  multipath fading conditions, for  a given destination bit error rate (BER) and maximum transmit power constraints.

In more details, the technical contributions of the present article are outlined below:

 \begin{itemize} 
\item[$\bullet$] Exact closed-form expressions are derived for the end-to-end SER of  $M{-}$QAM and $M{-}$PSK based  dual-hop regenerative relay networks with MRC reception at the destination over Nakagami${-}m$ multipath fading channels with arbitrary spatial correlation between source-destination (S-D) and relay-destination (R-D) links.

\item[$\bullet$] Simple closed-form asymptotic expressions are derived for the above scenarios for high SNR values.

\item[$\bullet$] The offered analytic results are employed in  a  comprehensive energy optimization analysis   based on minimizing the average total energy consumption of the overall regenerative relay network under a given QoS target, in terms of bit-error-rate (BER), and maximum transmit power constraints. 
\end{itemize}

The remainder of this paper is organized as follows: Section II presents the considered relay system and channel model while Sections III and IV are devoted to the derivation of the corresponding exact and asymptotic error rate results. The total power consumption   models are presented in Section V while the analysis of energy minimization and power allocation optimization based on the given constraints are provided in Section VI. Section VII   presents the corresponding numerical results along with extensive analysis and discussions while closing remarks are provided in Section VIII.          
    
\IEEEpubidadjcol


\section{System and Channel Model }

We consider a two-hop cooperative radio access system model consisting  of   source node (S), a relay node (R) and  a destination node (D), where each node is equipped with a single antenna, as illustrated in Fig. 1. Without loss of generality, the system can  represent both a conventional and emerging communication scenarios such as, for example,  a mobile ad-hoc network or a vehicle-to-vehicle communication system. The cooperative strategy  is based on a half-duplex DF relaying where  transmission is performed using  time division multiplexing. It is also assumed that  the destination is equipped with MRC reception and that information signals are subject to multipath fading conditions that follow the Nakagami${-}m$ distribution\footnote{It is noted that the considered system requires the least   resources in terms of  bandwidth and power compared to multi-relay assisted transmission and thus, it can be adequate for low complexity and low-energy wireless networks.}.

\begin{figure}[!t]
\centering{\includegraphics[keepaspectratio,width= 8cm]{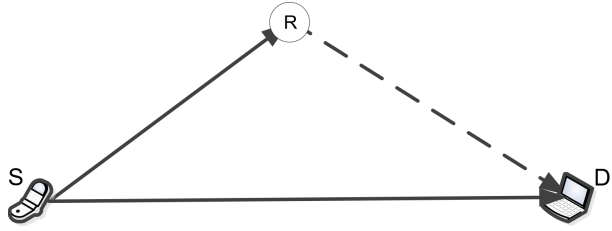}}
\caption{A dual-hop cooperative single relay model.} 
\end{figure}

In phase I, the source broadcasts the signal to both destination and relay nodes and the corresponding received signals can be expressed as 

\begin{equation}\label{L1}
y_{S,D} = \sqrt{\frac{P_{S}}{ P_{L_{S,D}}}} h_{S,D}x + n_{S,D}
\end{equation}
and

\begin{equation}\label{L2}
y_{S,R} = \sqrt{\frac{P_{S}}{P_{L_{S,R}}}} h_{S,R}x + n_{S,R}           
\end{equation}
respectively, where $P_S$ is the transmit power, $x$ is the transmitted symbol with normalized unit energy in the first transmission phase and $P_{L_{S,D}}$ and $P_{L_{S,R}}$ denote the path loss values in the S-D and source-relay (S-R) paths, respectively.  Also,  $h_{S,D}$ and $h_{S,R}$ are the complex fading  coefficients of the S-D and S-R wireless links, respectively,   whereas $n_{S,D}$ and $n_{S,R}$  are the corresponding complex Gaussian noise terms with zero mean and variance $N_{0}$.  The relay then checks whether the received signal can be decoded correctly, which can be, for example, realized by examining the included cyclic-redundancy-check (CRC) digits or the received SNR levels\cite{Add_4},\cite{D:Lee}. Based on this, if the signal is successfully decoded, the relay forwards it to the destination during phase II with power $ \overline{P}_{R} = P_{R}$; otherwise, the relay does not transmit and  remains idle with $ \overline{P}_{R} = 0$.  Hence, the signal at the destination during phase II can be represented as
   
\begin{equation}\label{L3}
y_{R,D} = \sqrt{\frac{ \bar{P}_{R}}{  P_{L_{R,D}}}} \: h_{R,D}x + n_{R,D} 
\end{equation}         
where $ \bar{P}_{R}$ is the transmit power of the relay, $P_{L_{R,D}} $ is the path loss of the R-D path and $h_{R,D}$ and $n_{R,D}$ denote the  channel coefficient and complex Gaussian noise term with zero mean and variance $N_{0}$, respectively. Finally, the destination combines the received direct and relayed signals based on MRC principle  where the combined SNR  can be expressed as follows \cite{D:Lee},\cite{HAS} 

\begin{equation}
 \gamma_{\rm MRC} =  \frac{(P_{S}/P_{L_{S,D}})\mid h_{S,D}\mid^2 +  (\bar{P}_{R}/P_{L_{R,D}})\mid h_{R,D}\mid^2}{N_{0}}
\end{equation}
The fading between the devices is assumed to follow the Nakagami${-}m$ distribution, which is a widely used model as the fading parameter  $m$ can easily account  for both severe and moderate fading conditions. Thus,  the corresponding channel  power gains ${\mid h_{S,D} \mid}^{2}$,${\mid h_{S,R}\mid}^{2}$ and ${\mid h_{R,D}\mid}^{2}$ follow the gamma distribution \cite{B:Nakagami} with different power  parameters, $\frac{1}{\Omega_{S,D}}$,$\frac{1}{\Omega_{S,R}}$, $\frac{1}{\Omega_{R,D}}$, and fading  parameters 
$m_{S,D}$, $m_{S,R}$ and $m_{R,D}$, respectively.  

 In the considered regenerative system,  arbitrary spatial correlation is assumed to exist between the S-D and R-D paths, as also adopted in the semi-analytical contribution of \cite{D: Lee}. To this effect, the corresponding  MGF for the case of Nakagami$-m$ fading is expressed as \cite{A:Simon}

\begin{equation} \label{MGF}
 M(s) = \left( 1- \frac{(\overline{\gamma}_{S,D}/P_{L_{S,D}} +\overline{\gamma}_{R,D}/P_{L_{R,D}} )}{m}s  + \frac{(1-\rho)\overline{\gamma}_{S,D}\overline{\gamma}_{R,D}/(P_{L_{S,D}}P_{L_{R,D}})}{m^2}    
s^{2} \right)^{-m}, \;s < 0 
\end{equation}
where $ \overline{\gamma}_{S,D} $ and $ \overline{\gamma}_{R,D}$ are  the corresponding average SNR values, $s$ denotes the MGF  parameter and $\rho = {\rm Cov}(|h_{S,D}|^{2}, |h_{R,D}|^{2}) {/} \sqrt{{\rm Var}(|h_{S, D}|^{2}) {\rm Var}(|h_{R,D}|^{2})} $ represents the involved  correlation coefficient \cite{B:Nakagami, A:Simon}, with ${\rm Cov}(\cdot)$ and ${\rm Var}(\cdot)$ denoting covariance and variance operations, respectively. It is acknowledged that other practical impairments, such as co-channel interference, are not considered in this work.

 \section{  SER   for $M-$QAM Modulation in Nakagami$-m$ Fading with Spatial Correlation}

This section is devoted to the error probability analysis of the dual-hop cooperative network over   Nakagami${-}m$  fading channels with spatial correlation between the direct and the relay-destination paths \cite{B:Nakagami}. To this effect and assuming MRC reception,   the  average end-to-end SER   can be expressed as   \cite{D:Lee}

\begin{equation} \label{L7}
\begin{split}
\overline{\rm SER}^{C}_{D}  =& F_{\rm QAM}\left[ \frac{1}{\left (1+\frac{(P_{S}\: /P_{L_{S,D}}) \, \Omega_{S,D}\, g_{QAM}}{N_{0}m_{c}\sin^{2}\theta}  \right )^{m_{c}}}\right] F_{\rm QAM}\left[\frac{1}{\left (1+\frac{(P_{S}/\: P_{L_{S,R}})\, \Omega_{S,R}\, g_{\rm QAM}}{N_{0}m_{S,R}\sin^{2}\theta}\right )^{m_{S,R}}}\right] \\
 &  + F_{\rm QAM}\left[\frac{1}{\left(1 + \frac{[(P_{S} \Omega_{S,D}/\: P_{L_{S,D}}) + ({P}_{R}\Omega_{R,D}/\: P_{L_{R,D}})]\: \: g_{QAM}}{N_{0}m_{c}\sin^{2}\theta}   + \frac{(1 - \rho) P_{S}  {P}_{R} \Omega_{S,D} \Omega_{R,D} g^{2}_{\rm QAM} }{N_{0}^{2}  P_{L_{S,D}}  P_{L_{R,D}} m_{c}^{2} \sin ^{4}\theta} \right)^{m_{c}}} \right] \\
  & \quad \times \left\lbrace 1- F_{\rm QAM}\left[ \frac{1}{\left(1 + \frac{(P_{S} /\: P_{L_{S,R}})\, \: \Omega_{S,R}\: g_{QAM}}{N_{0}m_{S,R}\sin^{2}\theta} \right )^{m_{S,R}}}\right] \right\rbrace  
\end{split}
\end{equation} 
 where $ m_{c} = m_{S,D} = m_{R,D}$, $ g_{\rm QAM} = {3}/{2(M -1)} $ and 

\begin{equation}\label{L8}
F_{\rm QAM}\left[ v(\theta) \right] = \frac{4}{\pi} \left( 1 - \frac{1}{\sqrt{M}} \right) \int_{0}^{\pi /2} v(\theta) {\: {\rm d} \theta} - \frac{4}{\pi} \left( 1 - \frac{1}{\sqrt{M}} \right)^{2} \int_{0}^{\pi /4} v(\theta) \: {\rm d} \theta.  
\end{equation}
The first two terms in \eqref{L7} refer to the cases of  incorrect and correct decoding of the received signal at the relay node, respectively, whereas the integral representation in \eqref{L8} is  used for evaluating the $\overline{\rm SER}^{C}_{D}$ numerically. 
In what follows, we firstly derive a closed-form expression for the average SER in the case of direct communication mode. This expression is subsequently employed in the derivation of   exact closed-form expressions for the average SER  of $M-$QAM and $M-$PSK modulated regenerative systems over Nakagami$-m$ fading channels with spatial correlation. Furthermore, it is used in the analysis of the energy consumption model and energy minimization in Sec. VI as it allows the derivation of an accurate expression for the energy consumption in the direct transmission, which  acts as a benchmark in the evaluation of  the energy reduction of the cooperative system.

\subsection{ Exact SER  for the Direct Transmission (DT)}

\begin{theorem}
For  $P_{S}, P_{L_{S,D}}, \Omega_{S,D}, N_{0}, g_{\rm QAM} \in \mathbb{R}^{+}$, $M \in \mathbb{N}$ and $m_{S, D} \geq \frac{1}{2}$, the symbol error rate of a  $M{-}QAM$    direct transmission scheme can be expressed as follows, 

\begin{equation} \label{D_new}
\begin{split}
\overline{\rm SER}^{D}_{D}  &= \frac{2(\sqrt{M} - 1) N_{0}^{m_{S,D}}m_{S,D}^{m_{S,D}} P_{L_{S,D}} }{\sqrt{\pi} M (m_{S,D} \, N_{0}  P_{L_{S,D}} + P_{S} \Omega_{S,D}  g_{\rm QAM})^{m_{S,D}}} \\
& \quad \times  \left\lbrace  \frac{\Gamma\left( m_{S,D} + \frac{1}{2} \right)}{\Gamma(m+1)} \, _{2}F_{1}\left( m_{S, D}, \frac{1}{2}, m_{S, D} + 1, \frac{m_{S, D} N_{0}  P_{L_{S,D}}}{m_{S, D} N_{0}  P_{L_{S,D}} + P_{S}\Omega_{S, D} g_{\rm QAM}} \right)  \right. \\
& \left. \, \qquad \,  +  \frac{\sqrt{2} (\sqrt{M} - 1)}{\sqrt{\pi}} F_{1} \left( \frac{1}{2}, \frac{1}{2} - m_{S, D}, m_{S,D}, \frac{3}{2}, \frac{1}{2}, \frac{m_{S, D} N_{0}  P_{L_{S,D}} }{2(m_{S, D}N_{0}  P_{L_{S,D}} + P_{S} \Omega_{S, D} g_{\rm QAM})}\right)\right\rbrace 
\end{split}
\end{equation}
\vspace{0.0cm}
where $\, _{2}F_{1}(.)$  and $F_{1}(.)$ denote the Gauss hypergeometric function and the Appell  hypergeometric function of the first kind, respectively. 
\end{theorem}
\begin{proof}
The proof is provided in Appendix A. 
\end{proof}


\subsection{Exact SER for the Cooperative-Transmission (CT)}

In this subsection, we derive a  novel closed-form expression for the average SER of the cooperative transmission scenario when the involved relay node decodes and forwards successfully decoded information signals to  the destination. To this end, it is essential to firstly derive exact closed-form expressions for two important indefinite trigonometric integrals.

\begin{lemma}
For $a, b, m \in \mathbb{R}^{+}$ and $2m  - \frac{1}{2} \in \mathbb{N}$, the following closed-form expression is valid, 

\begin{equation} \label{Lemma_1}
\begin{split}
\mathcal{J}(a, b,  m) &= \int \frac{1}{\left(1 +  \frac{a}{\sin^{2}(\theta)}  + \frac{b}{\sin^{4} (\theta)}\right)^{m}} {\rm d}\theta\\ 
& = - \sum_{l = 0}^{2m - \frac{1}{2}} \binom{2m - \frac{1}{2}}{l}  \,  (-1)^{l}  \, \frac{ \left(a+  2 \sin^{2}(\theta) - \sqrt{a^{2} - 4b} \right)^{m} \, \left(a + 2 \sin^{2}(\theta) + \sqrt{a^{2} - 4b} \right)^{m} }{\left(2 + a - \sqrt{a^{2} - 4b}\right)^{m} \left(2 + a + \sqrt{a^{2} - 4b}\right)^{m}  }  \\ 
& \quad \times  \frac{ \cos^{1+2l}(\theta) }{(1 + 2l) } \,  \frac{ F_{1}\left(l + \frac{1}{2}, m, m, l + \frac{3}{2}, \frac{2 \cos^{2}(\theta)}{2 + a - \sqrt{a^{2} - 4b}}, \frac{2 \cos^{2}(\theta)}{2 + a + \sqrt{a^{2} - 4b}} \right) }{ \left( \sin^{4}(\theta) + a \sin^{2}(\theta) + b \right)^{m}} + C. 
\end{split}
\end{equation}
\end{lemma}
\begin{proof}
The proof is provided in Appendix B. 
\end{proof}

\begin{lemma}
For $a, b, m, n \in \mathbb{R}^{+}$ and $m +  n - \frac{1}{2} \in \mathbb{N}$, the following closed-form expression is valid 

\begin{equation} \label{new_9}
\begin{split}
\mathcal{K}(a, b,  m, n) &= \int \frac{1}{\left( 1 + \frac{a}{\sin^{2}(\theta)} \right)^{m} \left( 1 + \frac{b}{\sin^{2}(\theta)} \right)^{n}} {\rm d} \theta \\
& = - \sum_{l=0}^{m + n - \frac{1}{2}} \binom{m + n - \frac{1}{2}}{l} \frac{(-1)^{l} \cos^{1 + 2l}(\theta) \, F_{1} \left(l + \frac{1}{2}, m, n, l + \frac{3}{2}, \frac{\cos^{2}(\theta)}{1+a}, \frac{\cos^{2}(\theta)}{1 + b} \right) }{(1 + 2l)(1 + a)^{m} (1 + b)^{n}} + C.   
\end{split}
\end{equation}
\end{lemma}
\begin{proof}
The proof is provided in Appendix C. 
\end{proof}

To the best of the authors' knowledge, the generic  solutions in the above Lemmas have not been previously reported in the open technical literature. These results are employed in the subsequent analysis.   
\begin{theorem}
For $\{ P_{S},  {P}_{R}, P_{L_{S,D}}, P_{L_{S,R}}, P_{L_{R,D}}, \Omega_{S,D}, \Omega_{S,R}, \Omega_{R,D},  N_{0}\} \in \mathbb{R}^{+}$, $M \in \mathbb{N}$, $ m_{S,D} \geq \frac{1}{2}$, $m_{S,R} \geq \frac{1}{2}$, $m_{R,D} \geq \frac{1}{2}$, $m_{S, D} $, $2m_{c} - \frac{1}{2} \in \mathbb{N}$ and $0 \leq \rho < 1$, the  SER of $M{-}$QAM based DF relaying over spatially correlated Nakagami${-}m$ fading channels, can be expressed as follows:  
 
\begin{equation*}
\overline{\rm SER}^{C}_{D}  = \left\lbrace \frac{2 \left(m_{c} - \frac{1}{2} \right)! \, _{2}F_{1}\left(m_{c}, \frac{1}{2}, m_{c} + 1, \frac{1}{1+a_1} \right)}{\sqrt{\pi}\, m_{c}! \, M (\sqrt{M} - 1)^{-1} \, (1 + a_{1})^{m_{c}}}  + \frac{ 2\sqrt{2}  \, F_{1}\left(\frac{1}{2}, \frac{1}{2} - m_{c}, m_{c}, \frac{3}{2}, \frac{1}{2},  \frac{1}{2+2a_{1}} \right)}{\pi \, M \,(\sqrt{M} - 1)^{-2}   \, (1 + a_{1})^{m_{c}}} \right\rbrace   
 \end{equation*}
\begin{equation} \label{SER_MQAM_Corr}
\begin{split}
& \times \left\lbrace \frac{2 \left(m_{S,R} - \frac{1}{2} \right)! \, _{2}F_{1}\left(m_{S,R}, \frac{1}{2}, m_{S,R} + 1, \frac{1}{1+b_{1}} \right)}{\sqrt{\pi}\, m_{S,R}! \, M (\sqrt{M} - 1)^{-1} \, (1 + b_{1})^{m_{S,R}}}  + \frac{  \, 2\sqrt{2}F_{1}\left(\frac{1}{2}, \frac{1}{2} - m_{S,R}, m_{S,R}, \frac{3}{2}, \frac{1}{2},  \frac{1}{2+2b_{1}} \right)}{  \pi \, M \,(\sqrt{M} - 1)^{-2}   \, (1 + b_{1})^{m_{S,R}}} \right\rbrace \\
& +  \left\lbrace 1 - \frac{2 \left(m_{S,R} - \frac{1}{2} \right)! \, _{2}F_{1}\left(m_{S,R}, \frac{1}{2}, m_{S,R} + 1, \frac{1}{1+b_{1}} \right)}{\sqrt{\pi}\, m_{S,R}! \, M (\sqrt{M} - 1)^{-1} \, (1 + b_{1})^{m_{S,R}}}  - \frac{   \, 2\sqrt{2}F_{1}\left(\frac{1}{2}, \frac{1}{2} - m_{S,R}, m_{S,R}, \frac{3}{2}, \frac{1}{2},  \frac{1}{2+2b_{1}} \right)}{ \pi \, M \,(\sqrt{M} - 1)^{-2}   \, (1 + b_{1})^{m_{S,R}}} \right\rbrace \\
& \times   \left\lbrace    \sum_{l = 0}^{2m_{c} - \frac{1}{2}} \binom{2m_{c} - \frac{1}{2}}{l} \frac{(-1)^{l} 4(\sqrt{M} - 1)^{2} 2^{m_{c} - l - \frac{1}{2}} F_{1}\left( l + \frac{1}{2}, m_{c}, m_{c}, l + \frac{3}{2}, 2\mathcal{A}, 2\mathcal{B} \right)}{M \pi (1 + 2l) [1 + 2(a_{1}+c_{1}) 4a_{1}d_{1}]^{m_{c}} \, [(1 - 2\mathcal{A}) (1 - 2\mathcal{B})]^{-m_{c}}} \right. \\
& \qquad + \left.   \sum_{l = 0}^{2m_{c} - \frac{1}{2}} \binom{2m_{c} - \frac{1}{2}}{l} \frac{(-1)^{l} 4 (\sqrt{M} - 1)  F_{1}\left( l + \frac{1}{2}, m_{c}, m_{c}, l + \frac{3}{2}, \mathcal{A}, \mathcal{B} \right) }{ \pi M  (1 + 2l) (a_{1} d_{1})^{m_{c}} (1 - \mathcal{A})^{-m_{c}} (1 - \mathcal{B})^{-m_{c}} }  \right\rbrace  
\end{split}
\end{equation}
where  $a_{1} = P_{S} \Omega_{S, D} g_{QAM}{/}(P_{L_{S,D}} m_{S,D} N_{0})$,  $ b_{1}= P_{S} \Omega_{S, R} g_{QAM}{/}(P_{L_{S, R}} N_{0} m_{S,R})$, $c_{1}= P_{R} \Omega_{R, D} g_{QAM}{/}$ $(P_{L_{R, D}}N_{0} m_{R, D})$, $d_{1} = (1 - \rho) P_{R} \Omega_{R, D} g_{QAM}{/}(P_{L_{R, D}}N_{0} m_{R,D})$  
and 

\begin{equation}
\big\{ ^{{\mathcal{A}}}_{{\mathcal{B}}} \big\} =  \frac{1}{2 + a_{1} + c_{1}  \, \left\lbrace ^{-}_{+} \right\rbrace \, \sqrt{(a_{1} + c_{1})^{2} - 4 a_{1} d_{1}}}. 
\end{equation}
\end{theorem}

\begin{proof}
The first term in \eqref{L7} corresponds to the direct transmission and thus, it can be expressed in closed-form based on Theorem 1. Likewise, the  second and the fourth term in  \eqref{L7}  have the same algebraic representation as \eqref{L10} and \eqref{L11} in Appendix A. Therefore,  they can be readily expressed in closed-form by making the necessary change of variables and substituting in  \eqref{trigonometric_c} and \eqref{trigonometric_e}.    As for the third term in \eqref{L7}, it is noticed that it has the same algebraic representation with \eqref{Lemma_1}.  As a result, a closed-form expression is deduced by determining the following specific cases in \eqref{Lemma_1}, which practically evaluate \eqref{L8}, 

\begin{equation}
\mathcal{J}\left(a, b,  m, 0, \left\lbrace^{\pi / 2}_{\pi / 4}  \right\rbrace  \right) = \int_{0}^{\left\lbrace ^{ \pi /2}_{\pi / 4} \right\rbrace } \frac{1}{\left(1 +  \frac{a}{\sin^{2}(\theta)}  + \frac{b}{\sin^{4} (\theta)}\right)^{m}} {\rm d}\theta. 
\end{equation}
Therefore,  by carrying out some   long but basic algebraic manipulations and substituting in \eqref{L7}  along with the aforementioned closed-form expressions, one obtains \eqref{SER_MQAM_Corr}, which completes the proof. 
\end{proof}

\begin{remark}
Equation \eqref{SER_MQAM_Corr} reduces to the  uncorrelated scenario by setting  $\rho = 0$. However, an alternative   expression for this case which is valid for the case that $\left\lbrace m_{S,D} + m_{R,D} - \frac{1}{2} \right\rbrace \in \mathbb{N}$ can also be  deduced by applying the derived expressions in Theorem 1 and Lemma 2 in \cite[eq. (11)]{D:Lee}, namely, 

\begin{equation*}
\begin{split}
\overline{\rm SER}^{C}_{D}  &= \left\lbrace \frac{2\Gamma\left(m_{S,D} + \frac{1}{2} \right) \, _{2}F_{1}\left(m_{S,D}, \frac{1}{2}, 1 + m_{S,D}, \frac{1}{1+a_1} \right) }{\sqrt{\pi} M  (\sqrt{M} - 1)^{-1}  \Gamma(1 + m_{S,D}) (1 + a_1)^{m_{S,D}}} - \frac{  2\sqrt{2} (F_{1}\left( \frac{1}{2}, \frac{1}{2}- m_{S,D}, \frac{3}{2}, \frac{1}{2}, \frac{1}{2 + 2a_1} \right) }{\pi (\sqrt{M} - 1)^{-2}  M (1+a_1)^{m_{S,D}}}\right\rbrace \\
& \times \left\lbrace \frac{2\Gamma\left(m_{S,R} + \frac{1}{2} \right) \, _{2}F_{1}\left(m_{S,R}, \frac{1}{2}, 1 + m_{S,R}, \frac{1}{1+b_1} \right) }{\sqrt{\pi} M  (\sqrt{M} - 1)^{-1}  \Gamma(1 + m_{S,R}) (1 + b_1)^{m_{S,R}}} - \frac{  2\sqrt{2} F_{1}\left( \frac{1}{2}, \frac{1}{2}- m_{S,R}, \frac{3}{2}, \frac{1}{2}, \frac{1}{2 + 2b_1} \right) }{\pi (\sqrt{M} - 1)^{-2}  M (1+b_1)^{m_{S,R}}}\right\rbrace  \\
\end{split}
\end{equation*}
\begin{equation}
\begin{split}
& +  \left\lbrace 1 - \frac{2\Gamma\left(m_{S,R} + \frac{1}{2} \right) \, _{2}F_{1}\left(m_{S,R}, \frac{1}{2}, 1 + m_{S,R}, \frac{1}{1+b_1} \right) }{\sqrt{\pi} M  (\sqrt{M} - 1)^{-1}  \Gamma(1 + m_{S,R}) (1 + b_1)^{m_{S,R}}} + \frac{ 2\sqrt{2} F_{1}\left( \frac{1}{2}, \frac{1}{2}- m_{S,R}, \frac{3}{2}, \frac{1}{2}, \frac{1}{2 + 2b_1} \right) }{\pi  (\sqrt{M} - 1)^{-2}  M (1+b_1)^{m_{S,R}}}\right\rbrace \\
&\quad \times \left\lbrace \sum_{l=0}^{m_{S,D} + m_{R,D} - \frac{1}{2}} \frac{4 (-1)^{l} (\sqrt{M} - 1) F_{1} \left(l + \frac{1}{2}, m_{S,D}, m_{R,D}, l + \frac{3}{2}, \frac{1}{1 + a_1}, \frac{1}{1 + c_1} \right) }{l! \pi M (1 + a_1)^{m_{S,D}} (1+c_1)^{m_{R,D}} (1 + 2l)  \left( m_{S,D} + m_{R,D} + \frac{1}{2} \right)_{-l}}  \right. \\
& \qquad  \quad + \left. \sum_{l=0}^{m_{S,D} + m_{R,D} - \frac{1}{2}} \frac{(-1)^{l} 2^{\frac{3}{2} - l} (\sqrt{M} - 1)^{2} F_{1} \left(l + \frac{1}{2}, m_{S,D}, m_{R,D}, l + \frac{3}{2}, \frac{1}{2 +2a_1}, \frac{1}{2 + 2c_1} \right) }{l! \pi M (1 + a_1)^{m_{S,D}} (1+c_1)^{m_{R,D}} (1 + 2l)  \left( m_{S,D} + m_{R,D} + \frac{1}{2} \right)_{-l}} \right\rbrace
\end{split}
\end{equation}
where  $(x)_{n} \triangleq \Gamma(x + n){/}\Gamma(x)$ denotes the Pochhammer symbol \cite{Tables}. 
\end{remark}

\subsection{Asymptotic SER for the Cooperative-Transmission }

Simple asymptotic expressions can be derived for the case of high SNR at the three paths of the system.  To this end, it is essential to firstly derive a  closed-form expression for another trigonometric integral. 

\begin{lemma}
For  $m \in \mathbb{R}$, the following generic closed-form expression holds,

\begin{equation} \label{sin_integral}
\int \sin^{2m} {\rm d}\theta = - \cos(\theta) \, _{2}F_{1}\left( \frac{1}{2}, \frac{1}{2} - m, \frac{3}{2}, \cos^{2}(\theta)  \right) + C. 
\end{equation}

\begin{proof}
The proof is provided in Appendix D. 
\end{proof}
\end{lemma}
 Lemma 3 is subsequently employed  in the derivation of the following proposition.  

\begin{proposition}
For $\{ P_{S},  {P}_{R}, P_{L_{S,D}}, P_{L_{S,R}}, P_{L_{R,D}}, \Omega_{S,D}, \Omega_{S,R}, \Omega_{R,D},  N_{0}\} \in \mathbb{R}^{+}$, $M \in \mathbb{N}$, $ m_{S,D} \geq \frac{1}{2}$, $m_{S,R} \geq \frac{1}{2}$, $m_{R,D} \geq \frac{1}{2}$, $m_{S, D} $, $2m_{c} - \frac{1}{2} \in \mathbb{N}$ and $0 \leq \rho < 1$, the  SER of $M{-}$QAM based DF relaying over spatially correlated Nakagami${-}m$ fading channels in the high SNR regime can be expressed as 

\begin{equation*}
\begin{split}
\overline{\rm SER}^{C }_{D}  &\simeq  \left( \frac{   P_{L_{S,D}} N_{0} m_{c}}{   P_{S} \Omega_{S,D} g_{QAM}}\right)^{m_{c}}    \left(1 - \frac{1}{\sqrt{M}} \right)  \left\lbrace \frac{2 \Gamma\left(m_{c} + \frac{1}{2} \right)}{  \sqrt{\pi} \Gamma(1 + m_{c})} -   \left(1 - \frac{1}{\sqrt{M}} \right)  \frac{_{2}F_{1}\left(\frac{1}{2}, 1, m_{c} + \frac{3}{2}, -1 \right)}{ \pi 2^{m_{c} - 2} \,   (1 + 2 m_{c})} \right\rbrace \\
& \times \left( \frac{   P_{L_{S,R}} N_{0} m_{S,R}}{  P_{S} \Omega_{S,R} g_{QAM}} \right)^{m_{S,R}}   \left(1 - \frac{1}{\sqrt{M}} \right)  \left\lbrace \frac{2\Gamma\left(m_{S,R} + \frac{1}{2} \right)}{ \sqrt{\pi} \, \Gamma(1 + m_{S,R})} -   \left(1 - \frac{1}{\sqrt{M}} \right)  \frac{_{2}F_{1}\left(\frac{1}{2}, 1, m_{S,R} + \frac{3}{2}, -1 \right)}{ \pi \,  2^{m_{S,R}-2} (1 + 2 m_{S,R})} \right\rbrace  
\end{split}
\end{equation*}
\begin{equation} \label{SER_as_QAM}
\begin{split}
&+ \left( \frac{  P_{L_{S,D}}P_{L_{R,D}} N^{2}_{0} m^{2}_{c}}{ (1 - \rho) \,  P_{S} P_{R} \Omega_{S,D} \Omega_{R,D} g^{2}_{QAM}}  \right)^{m_{c}}  \left(1 - \frac{1}{\sqrt{M}} \right)  \qquad \qquad \qquad   \qquad \qquad \qquad   \qquad   \\
& \qquad \times \left\lbrace \frac{2\Gamma\left(2m_{c} + \frac{1}{2} \right)}{\sqrt{\pi}  \, \Gamma(1 +2 m_{c})} -   \left(1 - \frac{1}{\sqrt{M}} \right)  \frac{_{2}F_{1}\left(\frac{1}{2}, 1, 2m_{c} + \frac{3}{2}, -1 \right)}{\pi \, 2^{2 m_{c} - 2} (1 + 4 m_{c})} \right\rbrace. 
\end{split}
\end{equation}
\end{proposition}
\begin{proof}
In the high SNR  regime, it is realistic to assume that $P_{S} \Omega_{S,D} >> P_{L_{S,D}}N_{0}$,  $P_{S} \Omega_{S,R} >> P_{L_{S,R}}N_{0}$ and $P_{R} \Omega_{R,D} >> P_{L_{R,D}}N_{0}$.  Based on this,  an integral representation was formulated in   \cite[eq. (28)]{D:Lee}, 

\begin{equation} \label{Asym_QAM}
\overline{\rm SER}^{C }_{D} \simeq   A_{c} A_{S,R}  \left( \frac{N_{0} m_{c}}{P_{S} \Omega_{S,D} g_{QAM}} \right)^{m_{c}}  \left( \frac{N_{0} m_{S,R}}{P_{S} \Omega_{S,R} g_{QAM}} \right)^{m_{S,R}}   +   A_{2c} \left( \frac{N_{0}^{2} m_{c}^{2} }{(1 - \rho) P_{S} P_{R} \Omega_{S,D}\Omega_{R,D}g_{QAM}^{2}}\right)^{m_{c}} 
\end{equation}
where

\begin{equation}
\left\lbrace ^{\, A_{c}}_{A_{2c}} \right\rbrace = \frac{4}{\pi} \left(1 - \frac{1}{\sqrt{M}} \right) \int_{0}^{\pi {/}2} \sin^{\left\lbrace ^{2m_{c}}_{4m_{c}} \right\rbrace} {\rm d}\theta - \frac{4}{\pi} \left(1 - \frac{1}{\sqrt{M}} \right)^{2} \int_{0}^{\pi {/}4} \sin^{\left\lbrace ^{2m_{c}}_{4m_{c}} \right\rbrace} {\rm d}\theta 
\end{equation} 
and

\begin{equation}
A_{S,R}  = \frac{4}{\pi} \left(1 - \frac{1}{\sqrt{M}} \right) \int_{0}^{\pi {/}2} \sin^{2 m_{S,R}} {\rm d}\theta - \frac{4}{\pi} \left(1 - \frac{1}{\sqrt{M}} \right)^{2} \int_{0}^{\pi {/}4} \sin^{2 m_{S,R}} {\rm d}\theta.  
\end{equation}
Evidently, the terms  $A_{c}$, $A_{2c}$ and $A_{S,R}$  can be expressed in closed-form with the aid of Lemma 3. Based on this, by performing the necessary change of variables in \eqref{sin_integral} and substituting in \eqref{Asym_QAM}, equation \eqref{SER_as_QAM} is deduced, which completes the proof.   
\end{proof}

\begin{remark}
Using \eqref{SER_as_QAM}, the correlation coefficient for the case of  M${-}$QAM modulation can be expressed in terms of the corresponding source and relay powers, fading parameters and average SER, namely

\begin{equation}\label{ccq}
  \rho =   1 - \frac{ K_{3}{\left\lbrace \frac{2C\Gamma\left(2m_{c} + \frac{1}{2} \right)}{\sqrt{\pi}  \, \Gamma(1 +2 m_{c})} -     \frac{C^{2}\,_{2}F_{1}\left(\frac{1}{2}, 1, 2m_{c} + \frac{3}{2}, -1 \right)}{\pi \, 2^{2 m_{c} - 2} (1 + 4 m_{c})} \right\rbrace 
   }^{\frac{1}{m_c}}}{\left(\overline{\rm SER}^{C }_{D} - \frac{\left\lbrace \frac{2 C\Gamma\left(m_{c} + \frac{1}{2} \right)}{  \sqrt{\pi} \Gamma(1 + m_{c})} -    \frac{C^{2} \,_{2}F_{1}\left(\frac{1}{2}, 1, m_{c} + \frac{3}{2}, -1 \right)}{ \pi 2^{m_{c} - 2} \,   (1 + 2 m_{c})} \right\rbrace 
\left\lbrace \frac{2C\Gamma\left(m_{S,R} + \frac{1}{2} \right)}{ \sqrt{\pi} \, \Gamma(1 + m_{S,R})} -     \frac{C^{2}\,_{2}F_{1}\left(\frac{1}{2}, 1, m_{S,R} + \frac{3}{2}, -1 \right)}{ \pi \,  2^{m_{S,R}-2} (1 + 2 m_{S,R})} \right\rbrace }{K_{1}^{-m_c}K_{2}^{-m_{S,R}} }\right)^{\frac{1}{m_c}}}
\end{equation}
where $g = g_{QAM}, C = \left( 1 - {1}/{\sqrt{M}}\right), {K_{1} = \left({N_{0} m_{c} P_{L_{S,D}}}/{P_{S} \Omega_{S,D} g} \right) }, 
K_{2} = \left({N_{0} m_{S,R} P_{L_{S,R}}}/{P_{S} \Omega_{S,R} g} \right)$,  and $K_{3} = \left({N_{0}^{2} m_{c}^{2} P_{L_{S,D}} P_{L_{R,D}}}/{ P_{S} P_{R} \Omega_{S,D} \Omega_{R,D} g^{2}} \right)$.
\end{remark}

 \section{ SER   for $M-$PSK Modulation in Nakagami$-m$ Fading with Spatial Correlation}
 
Having derived novel analytic expressions for the case of $M-$QAM modulation, this section is devoted to the derivation of exact and asymptotic  closed-form expressions  for the case of $M-$PSK constellations. 

\vspace{-0.33cm}

\subsection{Exact SER for the Cooperative-Transmission}

\begin{theorem}
For $\{ P_{S},  {P}_{R}, P_{L_{S,D}}, P_{L_{S,R}}, P_{L_{R,D}}, \Omega_{S,D}, \Omega_{S,R}, \Omega_{R,D},  N_{0}\} \in \mathbb{R}^{+}$, $M \in \mathbb{N}$, $ m_{S,D} \geq \frac{1}{2}$, $m_{S,R} \geq \frac{1}{2}$, $m_{R,D} \geq \frac{1}{2}$, $m_{S, D} $, $2m_{c} - \frac{1}{2} \in \mathbb{N}$ and $0 \leq \rho < 1$, the  SER of $M{-}$PSK based DF relaying over spatially correlated Nakagami${-}m$ fading channels, can be expressed as follows:
 
\begin{equation} \label{M-PSK_corr}
\begin{split}
\overline{\rm SER}^{C}_{D} &=   \frac{\sin^{2 m_{c} + 1}\left( \frac{(M-1)\pi}{M} \right)}{(1 + 2 m_{c}) \pi a_{2}^{m_{c}}} F_{1} \left( m_{c} + \frac{1}{2}, \frac{1}{2}, m_{c}, m_{c} + \frac{3}{2}, \sin^{2}\left( \frac{(M-1)\pi}{M} \right), \frac{\sin^{2}\left( \frac{(M-1)\pi}{M} \right)}{a_{2}} \right) \\
&  \times \frac{\sin^{2 m_{S,R} + 1}\left( \frac{(M-1)\pi}{M} \right)}{(1 + 2 m_{S,R}) \pi b_{2}^{m_{S,R}}} F_{1} \left( m_{S,R} + \frac{1}{2}, \frac{1}{2}, m_{S,R}, m_{S,R} + \frac{3}{2}, \sin^{2}\left( \frac{(M-1)\pi}{M} \right), \frac{\sin^{2}\left( \frac{(M-1)\pi}{M} \right)}{b_{2}} \right)  \\
&+ \left\lbrace 1 - \frac{\sin^{2 m_{S,R} + 1}\left( \frac{(M-1)\pi}{M} \right)}{(1 + 2 m_{S,R}) \pi b_{2}^{m_{S,R}}} F_{1} \left( m_{S,R} + \frac{1}{2}, \frac{1}{2}, m_{S,R}, m_{S,R} + \frac{3}{2}, \sin^{2}\left( \frac{(M-1)\pi}{M} \right), \frac{\sin^{2}\left( \frac{(M-1)\pi}{M} \right)}{b_{2}} \right) \right\rbrace
\end{split}
\end{equation}
\begin{equation*}
\begin{split}
&  \times \left\lbrace \sum_{l=0}^{2m_{c} - \frac{1}{2}}  \binom{2m_{c} - \frac{1}{2}}{l} \frac{(-1)^{l} F_{1} \left( l + \frac{1}{2}, m_{c}, m_{c}, l + \frac{3}{2}, \frac{2}{2 + c_{2} - \sqrt{c_{2}^{2} - 4 d_{2}}}, \frac{2}{2 + c_{2} + \sqrt{c_{2}^{2} - 4 d_{2}}}  \right) }{(1 + 2l) \pi d_{2}^{m_{c}} \left[ \left(1 - \frac{2}{2 + c_{2} - \sqrt{c_{2}^{2} - 4 d_{2}}} \right) \left(1 - \frac{2}{2 + c_{2} + \sqrt{c_{2}^{2}  - 4 d_{2}}} \right) \right]^{-m_{c}}} \right.  \\
& \quad \left. - \sum_{l=0}^{2m_{c} - \frac{1}{2}}  \binom{2m_{c} - \frac{1}{2}}{l} \frac{(-1)^{l} \cos^{1 + 2l}\left( \frac{(M-1) \pi}{M} \right) F_{1} \left( l + \frac{1}{2}, m_{c}, m_{c}, l + \frac{3}{2}, \frac{2 \cos^{2} \left( \frac{(M-1) \pi}{M} \right)}{2 + c_{2} - \sqrt{c_{2}^{2} - 4 d_{2}}}, \frac{\cos^{2}\left(\frac{(M-1)\pi}{M} \right)}{2 + c_{2} + \sqrt{c_{2}^{2} - 4 d_{2}}}  \right) }{(1 + 2l) \pi  \, \left(d_{2} \sin^{4}\left( \frac{(M-1) \pi}{M}\right) + c_{2} \cos^{2} \left( \frac{(M-1) \pi}{M} \right) + c_{2} \right)^{m_{c}} } \right.   \\
&  \qquad \quad \times \left.    \left(1 - \frac{2\cos^{2}\left(\frac{(M-1) \pi}{M} \right)}{2 + c_{2} - \sqrt{c_{2}^{2} - 4 d_{2}}} \right)^{m_{c}} \left(1 - \frac{2\cos^{2}\left(\frac{(M-1) \pi}{M} \right)}{2 + c_{2} + \sqrt{c_{2}^{2} - 4 d_{2}}} \right)^{m_{c}} \right\rbrace 
  \end{split}
\end{equation*}
where  $a_{2} = P_{S} \Omega_{S, D} g_{PSK}{/}(P_{L_{S,D}} m_{S,D} N_{0})$,  $ b_{2}= P_{S} \Omega_{S, R} g_{PSK}{/}(P_{L_{S, R}} N_{0} m_{S,R})$, $c_{2}= P_{R} \Omega_{R, D} g_{PSK}{/}$ $(P_{L_{R, D}}N_{0} m_{R, D})$, $d_{2} = (1 - \rho) P_{R} \Omega_{R, D} g_{PSK}{/}(P_{L_{R, D}}N_{0} m_{R,D})$. 
\end{theorem}

\begin{proof} 
As a starting point, the average SER of $M-$PSK modulated DF systems over Nakagami${-}m$ fading channels with spatial correlation can be formulated as  follows \cite[eq. (23)]{D:Lee} 
\begin{equation} \label{MPSK_Definition}
\begin{split}
 \overline{\rm SER}^{C}_{D}  &= F_{\rm PSK}\left[ \frac{1}{\left( 1 + \frac{P_{S} \Omega_{S,D} g_{\rm PSK}}{N_{0} m_{c} P_{L_{S,D}} \sin^{2}(\theta) } \right)^{m_{S,D}} } \right]  F_{\rm PSK}\left[ \frac{1}{\left( 1 + \frac{P_{S} \Omega_{S,R} g_{\rm PSK}}{N_{0} m_{c} P_{L_{S,R}} \sin^{2}(\theta) } \right)^{m_{S,R}} } \right]    \\
 &  +  F_{\rm PSK}\left[ \frac{1}{\left(1 + \frac{(P_{S}\Omega_{S,D} + P_{R} \Omega_{R,D})g_{\rm PSK} }{N_{0} m_{c} P_{L_{S,D}} P_{L_{R,D}} \sin^{2}(\theta)}  + \frac{(1 - \rho) P_{S} P_{R} \Omega_{S,D}\Omega_{R,D} g^{2}_{\rm PSK}}{N_{0}^{2} m_{c}^{2} P_{L_{S,D}} P_{L_{R,D}} \sin^{4}(\theta)}  \right)^{m_{c}}} \right] \\
 & \times 
 \left\lbrace 1 - F_{\rm PSK}\left[ \frac{1}{\left( 1 + \frac{P_{S} \Omega_{S,R} g_{\rm PSK}}{N_{0} m_{c} P_{L_{S,R}} \sin^{2}(\theta) } \right)^{m_{S,R}} } \right] \right\rbrace 
  \end{split}
\end{equation}
where
 
\begin{equation}
F_{\rm PSK}[u(\theta)] = \frac{1}{\pi} \int_{0}^{ \frac{(M - 1) \pi}{M}} u(\theta) {\rm d} \theta. 
 \end{equation}
The above  four integrals have the same algebraic form as  the  integrals in Theorem 1 and Lemma 1. Thus, the proof follows by performing the same necessary change of variables and substituting in \eqref{MPSK_Definition}. 
\end{proof}

\subsection{Asymptotic SER for the Cooperative-Transmission}

\begin{proposition}
For $\{ P_{S},  {P}_{R}, P_{L_{S,D}}, P_{L_{S,R}}, P_{L_{R,D}}, \Omega_{S,D}, \Omega_{S,R}, \Omega_{R,D},  N_{0}\} \in \mathbb{R}^{+}$, $M \in \mathbb{N}$, $ m_{S,D} \geq \frac{1}{2}$, $m_{S,R} \geq \frac{1}{2}$, $m_{R,D} \geq \frac{1}{2}$, $m_{S, D} $, $2m_{c} - \frac{1}{2} \in \mathbb{N}$ and $0 \leq \rho < 1$, the  SER of $M{-}$PSK based DF relaying over spatially correlated Nakagami${-}m$ fading channels in the high SNR regime can be expressed as
 
 \begin{equation*}
 \begin{split}
\overline{\rm SER}^{C}_{D} &\simeq  \left( \frac{N_{0} m_{c} P_{L_{S,D}}}{P_{S} \Omega_{S,D} g_{\rm PSK}} \right)^{m_{c}} \left\lbrace \frac{\Gamma\left( m_{c} + \frac{1}{2} \right)}{2 \sqrt{\pi} \, m_{c}!} + \frac{\cos\left(  \frac{\pi}{M}\right)\,_{2}F_{1} \left( \frac{1}{2}, \frac{1}{2} - m_{c}, \frac{3}{2}, \cos^{2}\left( \frac{\pi}{M} \right) \right) }{\pi} \right\rbrace  \\
& \times \left( \frac{N_{0} m_{S,R} P_{L_{S,R}}}{P_{S} \Omega_{S,R} g_{\rm PSK}} \right)^{m_{S,R}} \left\lbrace \frac{\Gamma\left( m_{S,R} + \frac{1}{2} \right)}{2 \sqrt{\pi} \, m_{S,R}!} + \frac{\cos\left(  \frac{\pi}{M}\right)\,_{2}F_{1} \left( \frac{1}{2}, \frac{1}{2} - m_{S,R}, \frac{3}{2}, \cos^{2}\left( \frac{\pi}{M} \right) \right) }{\pi} \right\rbrace 
\end{split}
 \end{equation*}

\begin{equation} \label{SER_MPSK_as}
 + \left( \frac{N_{0}^{2} m_{c}^{2} P_{L_{S,D}} P_{L_{R,D}}}{(1 - \rho)  P_{S} P_{R} \Omega_{S,D} \Omega_{R,D} g_{\rm PSK}^{2}} \right)^{m_{c}} \left\lbrace \frac{\Gamma\left( 2m_{c} + \frac{1}{2}\right)}{2 \sqrt{\pi} (2m_{c})! } + \frac{\cos\left(\frac{\pi}{M} \right) \, _{2}F_{1}\left( \frac{1}{2}, \frac{1}{2} - 2 m_{c}, \frac{3}{2}, \cos^{2} \left( \frac{\pi}{M} \right) \right) }{\pi} \right\rbrace.
\end{equation}
\end{proposition}

\begin{proof}
The asymptotic SER for high SNR values was formulated in \cite[eq. (27)]{D:Lee} 

\begin{equation} \label{MPSK_as}
\overline{\rm SER}^{C}_{D}  \simeq  \tilde{A}_{c}\tilde{A}_{S,R}  \left( \frac{N_{0} m_{c} P_{L_{S,D}}}{P_{S} \Omega_{S,D} g_{\rm PSK}} \right)^{m_{c}} \left( \frac{N_{0} m_{S,R} P_{L_{S,R}}}{P_{S} \Omega_{S,R} g_{\rm PSK}} \right)^{m_{S,R}}\, +  \tilde{A}_{2c} \, \left( \frac{N_{0}^{2} m_{c}^{2} P_{L_{S,D}}P_{L_{R,D}}}{(1-\rho) P_{S} P_{R} \Omega_{S,D} \Omega_{R,D} g^{2}_{\rm PSK} } \right)^{m_{c}} 
\end{equation}
where
 
\begin{equation} \label{MPSK_as_1}
\left\lbrace ^{\, \tilde{A}_{c}}_{\tilde{A}_{2c}} \right\rbrace = \frac{1}{\pi}   \int_{0}^{ \frac{(M-1) \pi }{M}} \sin^{\left\lbrace ^{2m_{c}}_{4m_{c}} \right\rbrace} {\rm d}\theta  
\end{equation} 
and

\begin{equation} \label{MPSK_as_2}
A_{S,R}  = \frac{1}{\pi}  \int_{0}^{ \frac{(M-1) \pi }{M}} \sin^{2 m_{S,R}} {\rm d}\theta.   
\end{equation}
Notably, the integrals in \eqref{MPSK_as_1} and \eqref{MPSK_as_2}  have the same algebraic representation as the  integral in Lemma 3. As a result, by performing the necessary change of variables and substituting in \eqref{MPSK_as}, one obtains  \eqref{SER_MPSK_as}, which completes  the proof. 
\end{proof}
\begin{remark}
 Based on \eqref{SER_MPSK_as}, the corresponding correlation coefficient  can be expressed in terms of the corresponding source and relay powers, fading parameters and average SER as

\begin{equation}
\rho = 1 -\frac{K_{3}\left\lbrace \frac{\Gamma\left( 2m_{c} + \frac{1}{2}\right)}{2 \sqrt{\pi} (2m_{c})! } + \frac{\cos\left(\frac{\pi}{M} \right) \, _{2}F_{1}\left( \frac{1}{2}, \frac{1}{2} - 2 m_{c}, \frac{3}{2}, \cos^{2} \left( \frac{\pi}{M} \right) \right) }{\pi} \right\rbrace^{\frac{1}{m_c}}}{
\left(\overline{\rm SER}^{C}_{D} - \frac{\left\lbrace \frac{\Gamma\left( m_{c} + \frac{1}{2} \right)}{2 \sqrt{\pi} \, m_{c}!} + \frac{\cos\left(  \frac{\pi}{M}\right)\,_{2}F_{1} \left( \frac{1}{2}, \frac{1}{2} - m_{c}, \frac{3}{2}, \cos^{2}\left( \frac{\pi}{M} \right) \right) }{\pi} \right\rbrace 
 \left\lbrace \frac{\Gamma\left( m_{S,R} + \frac{1}{2} \right)}{2 \sqrt{\pi} \, m_{S,R}!} + \frac{\cos\left(  \frac{\pi}{M}\right)\,_{2}F_{1} \left( \frac{1}{2}, \frac{1}{2} - m_{S,R}, \frac{3}{2}, \cos^{2}\left( \frac{\pi}{M} \right) \right) }{\pi} \right\rbrace}{K_{1}^{-m_c}K_{2}^{-m_{S,R}}}\right)^\frac{1}{m_c}}  
 \end{equation} 
where $g = g_{PSK}$ is set in the  $K_{1}, K_{2}$ and $K_{3}$ terms, which are given in Remark 2.
\end{remark}
Figure 2 illustrates the SER performance as a function of SNR  for 4${-}$QAM/QPSK modulations. The source-destination transmission distance is indicatively considered at 600m while the relay is assumed to be located in the middle and the transmit power is shared equally to the source and the relay. The corresponding path-loss effects are considered by adopting the path loss (PL)  model in \cite{SP} namely

\begin{equation}
 PL_{i,j}[dB] = 148 + 40\log_{10}(d_{i,j}[km])
\end{equation}
which has been shown to characterize adequately harsh communication scenarios and is particularly applicable to mobile relaying and device-to-device communications. It is clearly observed that the empirical simulated results  are in excellent agreement with the respective analytical results. Furthermore, the simple asymptotic results are also highly accurate at higher SNRs.

\begin{figure}[h!]
\centering{\includegraphics[keepaspectratio,width= 15cm]{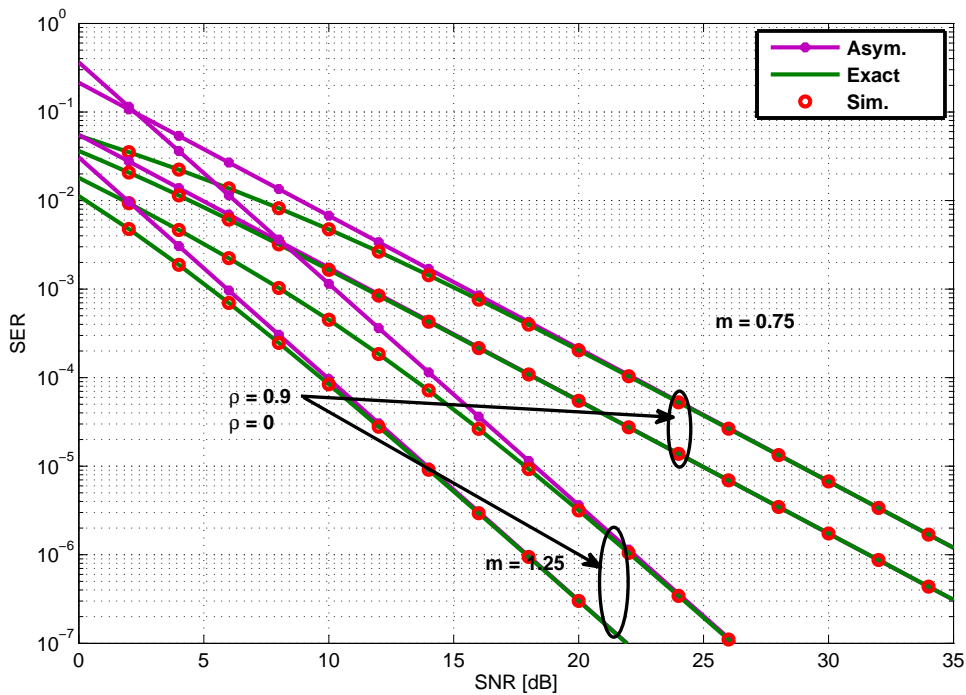} }
\caption{Example SER performance over Nakagami${-}m$ fading channels with  $m_{S,D} = m_{S,R} = m_{R,D} = m = \{0.75, 1.25\}, \, 
\Omega_{S,D} = \Omega_{S,R}  = \Omega_{R,D} = 0 $dB for $4{-}$QAM/QPSK constellations and different values of spatial correlation.  
}
\end{figure}

\section{System Power Consumption Model and Analysis}

In this section, motivated by the general interests towards green communications and increasing incentives to save energy, we quantify the total energy consumption required to transmit information from the source to the destination. We assume   that the transceiver circuitry operates on multi-mode basis i.e. $i)$ when there is a signal to transmit, the circuits are in active mode; $ii)$ when there is no signal to transmit, the circuits operate on a sleep mode; $iii)$ the circuits are in transient mode during the switching process from sleep mode to active mode. The elementary block diagrams of the  assumed transmitter and receiver   are illustrated in Fig. 2 and Fig. 3, respectively. This model is based on the energy and layout area efficient direct-conversion architecture which is commonly used in wireless transceivers. It is also assumed that all nodes are equipped with similar transmitter and receiver circuit blocks and that the power consumption of the active filters at the transmitter and receiver   is similar.

Considering a node that transmits $L$ bits and total transmission period $T$, the transient duration from active mode to sleep mode is short enough to be neglected. However, the start-up process from sleep mode to active mode may be slower due to the finite phase-locked loop (PLL) settling time in the frequency synthesizer. By denoting the duration of the sleep, transient and active modes as $T_{sp}$, $T_{tr}$ and $T_{on}$, respectively, the total transmission period is defined as $T = T_{sp} + T_{tr} + T_{on}$, with  $T_{tr}$ being equal to the frequency synthesizer settling time. Based on this, the total energy required to transmit and receive $L$ information bits is expressed as 

\begin{equation} \label{L32}
 E = P_{on}T_{on} + P_{sp}T_{sp} + P_{tr}T_{tr}
\end{equation}   
where $P_{on}, P_{sp}$ and $P_{tr}$ denote the power consumption values during the active, sleep and transient modes,  respectively. In realistic circuit designs, the power consumption in the sleep mode can be considered negligible compared to the active mode power \cite{Bahai} and thus, $P_{sp} \simeq  0$. It is also noted that power consumption during the transient mode practically refers to the power consumption of the frequency synthesizers. Based on this, it is assumed that $P_{tr} = 2P_{LO}$; therefore, using the power consumption values at both transmitter and receiver sides during the active mode one obtains
\begin{figure}[tp!]
\centering{\includegraphics[keepaspectratio,width= 15cm]{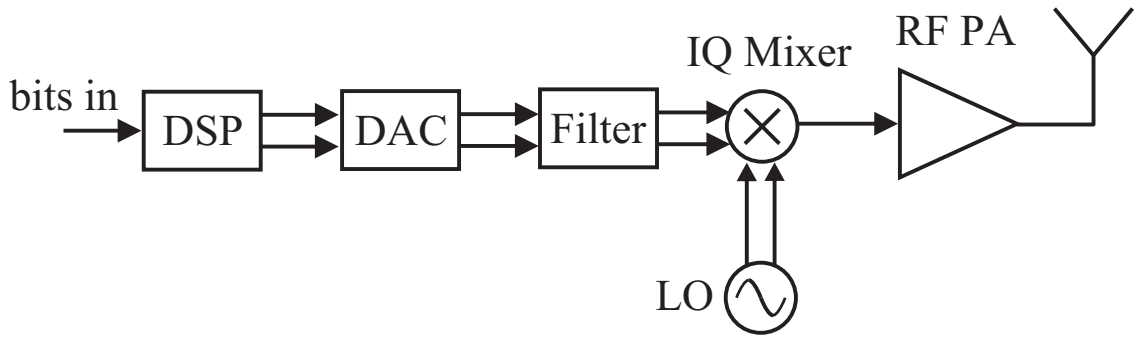}}
\caption{Elementary direct-conversion transmitter.} 
\end{figure}
\begin{figure}[tp!]
\centering{\includegraphics[keepaspectratio,width=  15cm]{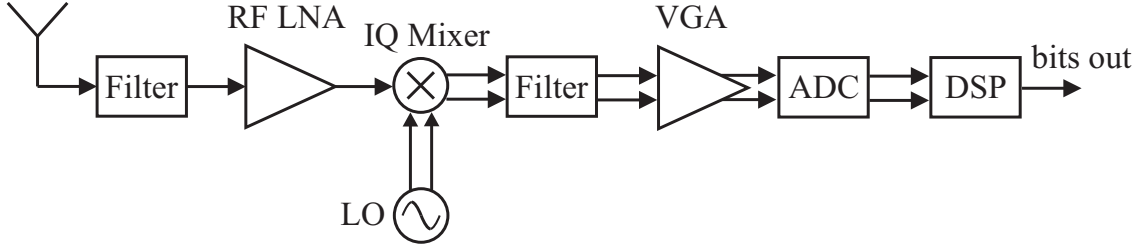}}
\caption{Elementary direct-conversion receiver.} 
\end{figure}
  
\begin{equation} \label{L33}
P_{on} = P_{ont} + P_{onr}
\end{equation}
where $P_{ont}$ is the total transmitter power consumption that accounts for the sum of signal transmission and transmitter circuit powers and $P_{onr}$ is the total receiver power consumption. Hence, it   follows that,
 
\begin{equation}  \label{L34}
 P_{ont} = P_{t} + P_{amp.} + P_{CT_x} 
\end{equation}
and
 
\begin{equation} \label{L35}
 P_{onr} = P_{CR_x} 
\end{equation}
 where $P_{t}$ is the signal transmission power, $P_{amp}$ is the power consumption of the RF power amplifier and $P_{CT_x}$ and $P_{CR_x}$ denote the total transmitter circuit power  and the total receiver circuit powers, namely, 
 
\begin{equation}  \label{L36}
P_{CT_{x}} = P_{DSP_{Tx}} + P_{DAC} + P_{Fil} + P_{Mix} + P_{LO} 
\end{equation}
and
 
\begin{equation} \label{L37}
P_{CR_{x}} = P_{DSP_{Rx}} + P_{ADC} +  P_{VGA} + 2P_{Fil} + P_{Mix} + P_{LO} + P_{LNA}  
\end{equation}
respectively. The $P_{CT_{x}} $ measure consists of the following power consumption entities: digital signal processor (DSP), $P_{DSP_{Tx}} $; digital to analog converter (DAC), $P_{DAC}$; active filter,  $ P_{Fil}$; IQ Mixer, $P_{Mix}$ and synthesizer, $P_{LO}$. Likewise, the active power consumption at the receiver comprises the power consumption values for  digital signal processor (DSP), $ P_{DSP_{Rx}}$; analog to digital converter(ADC), $P_{ADC}$; variable gain amplifier (VGA), $ P_{VGA}$; active filter, $P_{Fil}$; IQ Mixer, $P_{Mix}$; synthesizer, $P_{LO}$;  low noise amplifier (LNA), $P_{LNA}$  \cite{W:Amine}. Based on this,  the   total required circuit power consumption  is given by 
 
\begin{equation}  \label{L38}
P_{TC} = P_{CT_x} + P_{CR_x}.
\end{equation}

It is also noted that for signal transmission power $P_t$, the power consumption of the RF-power amplifier can be modeled by $P_{amp} = \alpha P_t$, where $ \alpha = \frac{\xi}{\eta}-1$, with $\eta$ and $\xi$  denoting the respective drain efficiency  of the amplifier and  the peak-to-average power ratio (PAPR), which depends on the modulation order and the associated constellation size. Based on this, for the case of square uncoded $M{-}$QAM modulation, $\xi = 3\frac{\sqrt{M}-1}{\sqrt{M}+1}$ and $T_{on} = \frac{LT_{s}}{b} = \frac{L}{bB}$, where $b= \log_2 M$ is the constellation size, $ L $ is  the transmission block length in bits \,and\, $T_{s}$ is  the symbol duration which   relates to the bandwidth $ B $ as $T_{s} \approx \frac{1}{B}$ \cite{X:Goldsmith}.

\section{Energy Optimization and Power Allocation}

In this section,  we deploy and combine the results of the previous sections and analyze the total energy required to transmit information efficiently from the source to the destination. To this end, we  firstly quantify the total energy consumption in the direct communication scenario. Hence, by applying \eqref{L32}, \eqref{L34} and \eqref{L35} and recalling that $P_{sp} \approx 0 $ and $P_{tr} = 2P_{LO}$, the average energy consumption per information bit   is given by \cite{GXX} 

\begin{equation}  \label{L39}
E^{D}_{T}  =  \overline{E}^{D}_{T} = \frac{\left((1 +\alpha )P_{S} + P_{CT_x} + P_{CR_x}\right)T_{on} + 2P_{LO}T_{tr}}{L} 
\end{equation} 
 
 \noindent 
 where $P_{S}$ denotes the source transmit power. In order to determine  the average total energy consumption   in the corresponding cooperative transmission system deploying the DF protocol,   we  formulate the total average   power consumption, which is a discrete random variable that can be statistically expressed as

\begin{equation}  \label{L40}
P_{T}^{C} = 
\begin{cases}
 & P_{CT_x}+  \left(1+\alpha\right)P_{S} + 2P_{CR_x},
   \hspace{1.075cm} \text{with} \, \,\,\,   {\rm Pr} = 1 \\
 & P_{CT_x}+  \left(1+\alpha\right)P_{R} +  P_{CR_x},  \hspace{1.2cm} \text{with }  \, \, {\rm Pr} =  1- \overline{\rm SER}_{S,R}
\end{cases}
\end{equation}
where $ P_{R}$ denotes the relay transmit power. The  first term of \eqref{L40} refers to the absolute total power consumption by the nodes in the first transmission phase, while the second term represents the  power consumption in the second phase, subject to correct decoding of the received signal by the relay, which is indicated  by the probabilistic term  $ \left(1- \overline{\rm SER}_{S,R}\right)$. Hence, the average total power consumption in the cooperative transmission mode can be expressed as  

\begin{equation}  \label{L41}
\bar{P}^{C}_{T}  =   P_{CT_x} + (1 + \alpha)P_{S} + 2P_{CR_x} \; + \;
 \left( P_{CT_{x}} + (1 + \alpha)P_{R} + P_{CR_{x}} \right) \left(1- \overline{\rm SER}_{S,R}\right).  
\end{equation}
Based on this, the corresponding average energy consumption per information bit is given by 

\begin{equation}  \label{L42}
\bar{E}^{C}_{T} =\frac {\overline{P}^{\, C}_{T} T_{on} + 2P_{LO}T_{tr}}{L}. 
\end{equation}

The achieved energy efficiency enhancement by the cooperative transmission is determined with the aid of the cooperation gain (CG), which is the  ratio of the energy efficiency of cooperative transmission over the energy efficiency of the direct transmission, per successfully delivered bit, namely, 

\begin{equation}\label{ss}
CG = \frac{\overline{E}^{D}_{T}\left(1-\overline{\rm BER}^{C}_{D}\right)}{\overline{E}^{C}_{T}\left(1-\overline{\rm BER}^{D}_{D}\right)}.
\end{equation}
Evidently, when the resulting ratio is smaller than one, it indicates that direct transmission is more energy efficient and thus, the extra energy consumption induced by cooperation outweighs its gain in decreasing the BER of the system.  

In what follows, the above expressions are employed  in formulating and solving the energy optimization problems aiming to guarantee certain  QoS requirements, namely, target destination BER.  
 In the same context, we additionally provide the optimal power allocation formulation for the cooperative transmission scenario under the maximum total transmit power constraint.
    
\subsection{Direct Transmission}

We first consider the energy optimization problem for minimizing the average total energy consumption in the direct communication scenario with the maximum transmission power and  target bit error rate, $ p^{*}$, as  constraints.   We assume that the power consumption of the circuit components are fixed and independent of the optimization. Thus, the only variable in the optimization is the transmit power of the source.  To this effect and with the aid of \eqref{L39}, the optimization problem for the direct transmission mode can be formulated as follows:    

\begin{equation}\label{L44}
\begin{split}
\min_{P_{S}}\overline{E}^{D}_{T} \hspace{4cm} \\ 
\quad \text{subject to:}\: \: \:  P_{S} \leq P_{maxt}, \quad   P_{S} \geq 0 \hspace{2cm} \\ 
\overline{\rm BER}^{D}_{D} =  p^{*}. \hspace{3.5cm}
\end{split}  
\end{equation} 
Deriving the minimum average total energy required in the direct communication scenario, requires prior computation of the corresponding  symbol error probability. This is  realized with the aid of \eqref{D_new} which is expressed in  closed-form in terms of $ _{2}F_{1} \left(m, \frac{1}{2}; m + 1; \frac{1}{1 + a_{1}} \right)$ and 
$F_{1}\left(\frac{1}{2};\frac{1}{2} - m, m, \frac{3}{2};\frac{1}{2},\frac{1}{2 + 2a_{1}}\right)$ functions\footnote{For the sake of simplicity, we assume that $m = m_{S,D}$. }. It is recalled that these functions are widely employed in  natural sciences and engineering and their computational implementation is rather straightforward as they  are   built-in functions in popular software packages such as MATLAB, MAPLE and MATHEMATICA. It is also noted that the representation of these functions in the present analysis   allows the following useful approximative expressions:  
$ _{2}F_{1} \left(m, \frac{1}{2}; m + 1; \frac{1}{1 + a_{1}} \right) \simeq 1 $ and  $ F_{1} \left( \frac{1}{2}; \frac{1}{2} - m, m, \frac{3}{2}; \frac{1}{2}, \frac{1}{2 + 2a_{1}}\right) \simeq  F_{1} \left( \frac{1}{2}; \frac{1}{2} - m, m, \frac{3}{2}; \frac{1}{2},0 \right)$. The accuracy of these approximations   is validated through extensive numerical and simulation results which indicate their tightness for random values of $m$ and moderate and large values of $a_1$. To this effect,    the following accurate  closed-form BER approximation for $M{-}$QAM modulated signals is deduced

\begin{equation}  \label{L45a}
\overline{\rm BER}^{D}_{D} \simeq    \frac{2 \; (\sqrt{M}-1)\Gamma( m + \frac{1}{2})}{ \sqrt{\pi} M m!(1 + a_1)^m  \log_2 M} +  \frac{4 \, F_{1}(\frac{1}{2};\frac{1}{2}-m,m,\frac{3}{2};\frac{1}{2}, 0)}{\sqrt{2} \pi (1 +a_1 )^{m} \log_2 M} \left(1-\frac{1}{\sqrt{M}}\right)^{2}.    
\end{equation} 
Importantly, the Appell  function in \eqref{L45a} can be expressed in terms of the Gauss hypergeometric function, 

\begin{equation} \label{identity}
F_{1}\left(\frac{1}{2};\frac{1}{2}-m,m,\frac{3}{2};\frac{1}{2}, 0\right) =  \,_{2}F_{1}\left( \frac{1}{2}, \frac{1}{2} -m; \frac{3}{2}, \frac{1}{2} \right). 
\end{equation}
As a result, equation \eqref{L45a} becomes 

\begin{equation}  \label{L45}
\overline{\rm BER}^{D}_{D} \simeq    \frac{2 \; (\sqrt{M}-1) \Gamma( m + \frac{1}{2})}{ \sqrt{\pi} M m! (1 + a_1)^m  \log_2 M} +  \frac{4 \,\,_{2}F_{1}\left( \frac{1}{2}, \frac{1}{2} - m; \frac{3}{2}, \frac{1}{2} \right)}{ \sqrt{2} \pi (1 +a_1 )^{m} \log_2 M} \left(1-\frac{1}{\sqrt{M}}\right)^{2}. 
\end{equation}
It is evident that \eqref{L45} is a function of the modulation order, the severity of multipath fading and $a_1$. Therefore, by substituting the targeted QoS  $p^{*}$ in \eqref{L45}, recalling that $a_1 =  (P_{S}\: \: \Omega_{S,D}\: g_{QAM}) {/} (N_{0}P_{L_{S,D}}) \:$ and carrying out some algebraic manipulations,  one obtains
  
\begin{equation}\label{La46}
P_{S} \simeq    \frac{m N_{0}  P_{L_{S,D}}}{\Omega_{S,D}g_{QAM}} \left[ \left(\frac{C}{p^{*}}\right)^{1/m} - 1\right]
\end{equation} 
where

\begin{equation}\label{L47}
C = \frac{4 (\sqrt{M} - 1)}{  M \log_2 M} \left[  \frac{(\sqrt{M} - 1)  \,_{2}F_{1}\left(\frac{1}{2}, \frac{1}{2} - m; \frac{3}{2}; \frac{1}{2} \right) }{\sqrt{2}\: \pi } + \frac{ \; \Gamma( m + \frac{1}{2})}{2 \sqrt{ \pi}m\Gamma(m)} \right]. 
\end{equation} 
To this effect and with the aid of \eqref{L39} and \eqref{La46}, it follows that the minimum total energy per information bit required for direct transmission for meeting the required QoS can be expressed in closed-form as 

\begin{equation}\label{L48}
\overline{E}^{D*}_{T} = \frac{\left(P_{CT_x} + P_{CR_x}  \right) T_{on} }{L} +    \frac{(1 + \alpha) N_{0} m_{S, D} T_{on} P_{L_{S,D}} }{L\, \Omega_{S,D}g_{QAM}} \left[\left(\frac{C}{p^{*}}\right)^{1/m} - 1\right]    + \frac{  2P_{LO}T_{tr}}{L}.
\end{equation}

Based on the total energy consumption in \eqref{L48} and given  the constellation size $ b = \frac{L}{B T_{on}} $, it is shown that the proposed energy expression comprises the  transmission energy $ E_{t} $ and circuit energy $ E_{C} $, namely, 

\begin{equation}  \label{L49}
E_{t} = P_{S}T_{on} 
      = \frac{N_{0}m PL_{S,D}}{\Omega_{S,D}g_{QAM}} \left[ \left(\frac{C}{p^{*}}\right)^{1/m} - 1\right] \frac{T_{on}}{L}
\end{equation}     
where $C$ can be expressed as a function of the transmission time $T_{on}$ as follows
 
\begin{equation} \label{L50}
C =  \frac{  B T_{on} \, \left( 2^\frac{L}{2B T_{on}} - 1 \right)^{2}}{  L \pi 2^{\frac{L}{B T_{on}} -\frac{3}{2}} } \, \, _{2}F_{1} \left(\frac{1}{2}, \frac{1}{2} - m; \frac{3}{2};\frac{1}{2} \right) + \frac{ 2 B T_{on}  }{ 2^{\frac{L}{BT_{on}}}m L \sqrt{\pi} }\, \left(\frac{1}{2}\right)_{m}  \left( 2^\frac{L}{2B T_{on}} - 1 \right).
\end{equation}   
Hence, by inserting \eqref{L50} in \eqref{L49},   the following analytic expression  for the transmission energy per information bit is deduced  
 
\begin{equation}  \label{L51}
\hspace{-1cm} E_{t} =  \frac{N_{0}m PL_{S,D} T_{on}}{L \, \Omega_{S,D}g_{QAM}} \left\lbrace\left[  \frac{  B T_{on} \, \left( 2^\frac{L}{2B T_{on}} - 1 \right)^{2}}{ p^{*} L \pi 2^{\frac{L}{B T_{on}} -\frac{3}{2}} } \, \, _{2}F_{1} \left(\frac{1}{2}, \frac{1}{2} - m; \frac{3}{2};\frac{1}{2} \right) + \frac{  B T_{on}\left( 2^\frac{L}{2B T_{on}} - 1 \right)}{p^{*}m L \sqrt{\pi}\;  2^{\frac{L}{BT_{on}}-1}  }\, \left(\frac{1}{2}\right)_{m}\right]^{1/m} - 1 \right\rbrace.   
\end{equation}      
Likewise, the circuit energy, $E_C$, can be expressed as  

\begin{equation}  \label{L52}
E_{C} = \left(P_{CT_x} + P_{CR_x}\right)T_{on}. 
\end{equation}
Notably, equations  \eqref{L50} and \eqref{L51} indicate that for a fixed bandwidth $B$ and packet length $L$, the transmission energy is a decreasing function with respect to the product $T_{on}B$ whereas the circuit energy increases monotonically with respect to $T_{on}$. In addition, it is shown that the transmission energy is dependent upon  the transmission distance $d_{S,D}$ and the severity of fading $m$, while the corresponding circuit energy remains fixed regardless of the value of $d_{S,D}$ and $m$.

\subsection{Cooperative Transmission}

This subsection presents the energy optimization  and power allocation problem   when the involved relay   forwards successfully decoded signals, generally at different power than the power of the source.  Evidently, the respective optimization model is a two dimensional problem and thus, we formulate the energy minimization problem with two optimization variables, namely, the source transmit power, $P_{S}$, and the relay transmit power, $ {P}_{R}$. In this context, the  aim is to minimize the total energy consumption of the overall network instead of minimizing the energy consumption at individual nodes. Based on this and with the aid of   \eqref{L42}, the optimization problem can be  formulated as follows:           

  \begin{equation}  \label{L53}
 \begin{split}
 \hspace{4cm} \min_{P_{S}, {P}_{R}}\overline{E}^{C}_{T}(P_{S}, {P}_{R}) \hspace{4cm} \\ 
\quad \text{subject to:}\quad (P_{S} + {P}_{R} )  \leq P_{maxt}, \quad P_{S} \geq  0, \quad  {P}_{R} \geq 0  \hspace{2cm} \\ \vspace{0.5cm}  
 \overline{\rm BER}^{C}_{D}(P_{S}, {P}_{R}) =  p^{*}. \hspace{3.5cm}  
 \end{split}
 \end{equation} 
The above optimization task   is a non-linear programming problem since the objective function and the constraint BER are both non-linear functions of $P_{s}$ and  
 $ {P}_{R}$. It is also recalled that Karush Kuhn Tucker (KKT) conditions that handle both equality and inequality constraints are in general the first order sufficient and  necessary conditions for optimum solutions in non-linear optimization problems (NLP)  provided that certain regularity conditions are satisfied. To this end, using the Lagrange multipliers $\lambda_{1}$ and  $\lambda_{2}$, $\lambda_{3}$ and $\lambda_{4}$, for the equality and inequality constraints  we set the corresponding Lagrangian equation that depends   on the optimization variables and   multipliers while meeting the KKT conditions in \cite{body} for the non-linear convex  optimization problem, namely      

\begin{equation}  \label{L54}
L  =  \overline{E}^{C}_{T} + \lambda_{1} \left(\overline{\rm BER}^{C}_{D}- p^{*}\right) - \lambda_{2}P_{S} - \lambda_{3} {P}_{R} + \lambda_{4} \left(\left(P_{S} + {P}_{R} \right) - P_{maxt} \right).  
\end{equation} 
The proof for the convexity of the optimization problem is provided in   Appendix E. 

Based on \eqref{L54}, the KKT conditions for the problem can be expressed as follows:     

\begin{equation}  \label{L55}
\nabla\overline{E}^{C}_{T} + \lambda_{1}\nabla\overline{\rm BER}^{C}_{D} - \lambda_{2}\nabla P_{S} - \lambda_{3}\nabla {P}_{R} + \lambda_{4}\nabla\left(P_{S} +  {P}_{R}\right) = 0  
\end{equation}
whereas the associated complementary conditions are given by

 \begin{equation}  \label{L56}
\begin{split}
\overline{\rm BER}^{C}_{D} =  p^{*}, \, P_{S} +  {P}_{R} \leq P_{maxt} \hspace{3cm} \\ 
\lambda_{1}\left(\overline{\rm BER}^{C}_{D}- p^{*}\right) = 0, \, \, \lambda_{2}P_{S} = 0, \, \,  \lambda_{3} {P}_{R} = 0\,\: \, \text{and} \, \: \lambda_{4}\left(P_{S} +  {P}_{R} - P_{maxt}\right) = 0 \\
\lambda_{1}, \lambda_{2}, \lambda_{3}, \lambda_{4} \geq 0.\hspace{5cm}
\end{split}
\end{equation}  
In the above set of complementary KKT conditions both $\lambda_{2}$ and $\lambda_{3}$ represent inactive constraints  and therefore, they can be assumed  zero.  To this effect, by applying \eqref{L55} and   setting the derivatives w.r.t $P_{S}$ and $  {P}_{R}$ to zero,  the following  useful set of equations is deduced   
  
\begin{equation}  \label{L57}
\frac{\partial\overline{E}^{C}_{T}}{\partial P_{S}}\:   + \, \, \lambda_{1}\frac{\partial\overline{\rm BER}^{C}_{D}}{\partial P_{S}} \, \, + \lambda_{4} = 0 
\end{equation}
and

\begin{equation} \label{L58}
 \frac{\partial\overline{E}^{C}_{T}}{\partial {P}_{R}}\:   + \, \, \lambda_{1}\frac{\partial\overline{\rm BER}^{C}_{D}}{\partial  {P}_{R}} \, \, + \lambda_{4} = 0. 
\end{equation}
Solving for $ \lambda_{4} $ from \eqref{L57} and substituting in \eqref{L58} yields the following relationship which depends only on one of the Lagrangian multipliers
 
\begin{equation*}
\frac{\partial\overline{E}^{C}_{T}}{\partial P_{S}}\:-\frac{\partial\overline{E}^{C}_{T}}{\partial  {P}_{R}}\:   + \, \, \lambda_{1}\left(\frac{\partial\overline{\rm BER}^{C}_{D}}{\partial P_{S}}-\frac{\partial\overline{\rm BER}^{C}_{D}}{\partial  {P}_{R}}\right) = 0
\end{equation*} 
and

\begin{equation}
\lambda_{1} = \frac{\frac{\partial\overline{E}^{C}_{T}}{\partial P_{S}}\:-\frac{\partial\overline{E}^{C}_{T}}{\partial  {P}_{R}}} {\left(\frac{\partial\overline{\rm BER}^{C}_{D}}{\partial   {P}_{R}}-\frac{\partial\overline {\rm BER}^{C}_{D}}{\partial P_{S}}\right)}. 
\end{equation} 
Based on this and using the fact that $\lambda_{1} \geq 0 $, one obtains the following necessary condition for minimizing the total average energy consumption of the cooperative transmission mode at the optimal power values

\begin{equation}\label{L59}
\frac{\partial\overline{E}^{C}_{T}(P_{S}^{\, *},P_{R}^{\, *})}{\partial P_{S}}\,\geq\, \frac{\partial\overline{E}^{C}_{T}(P_{S}^{\, *},P_{R}^{\, *})}{\partial  {P}_{R}}.  
\end{equation}
For a feasible set of optimal powers, the    $\overline{\rm BER}^{C}_{D} =   p^{*}$  and  $P_{S} +  {P}_{R} \leq P_{maxt} $ constraints must be satisfied.

Analytic  solution for the optimal powers in \eqref{L59}  is intractable   to derive in closed form. However, this can be alternatively realized with the aid of numerical optimization techniques, which can determine the optimal powers at the source and relay nodes that minimize the average total energy consumption. To this end, we employ the MATLAB optimization tool box and its function \textit{fimincon} in the respective numerical calculations for allocating the available power optimally under the given constraints. 
Thus, the derived expressions and offered results provide tools to understand, quantify and analyze how much energy can in general be saved, per successfully communicated bit in the system, if transmit power allocation and optimization beyond classical equal power allocation is pursued in the cooperative system, on one side, and how much energy can be saved against the classical non-cooperative (direct transmission) system, on the other side.  Furthermore,  the considered values in the present paper are indicative and are selected in the context of demonstrating the validity of the proposed method. Therefore, the derived optimization flow can be readily extended to  arbitrary design constraints for the total network power consumption and target destination error rate in the presence of Nakagami$-{m}$  multipath fading conditions.  

\begin{figure}[htp!]
\centering{\includegraphics[ keepaspectratio,width= 15cm]{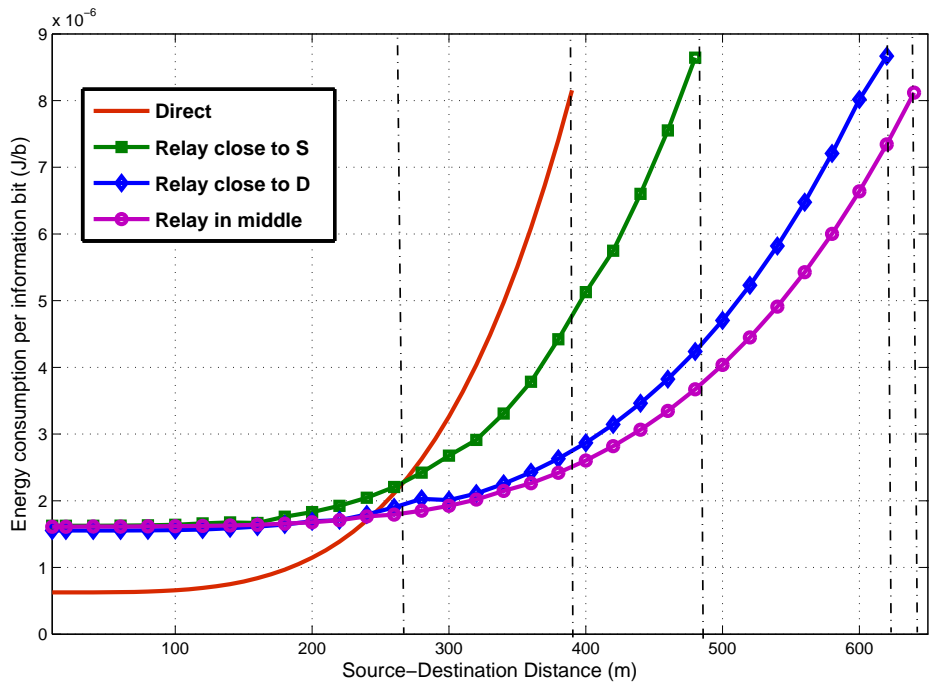} }
\caption{ Energy consumption per information bit versus source-destination distance for different relay locations over uncorrelated Nakagami${-}1.25$ at target BER of  $ 10^{-2}$ for $4{-}$QAM / QPSK constellation.} 
\end{figure}

\section{Numerical Results and Analysis}

This section demonstrates and evaluates the average total energy consumption of the considered regenerative  system  assuming that the S-D and S-R links are statistically independent whereas the S-D and R-D  paths  are spatially correlated.  As a realistic example, we assume $M{-}$QAM modulation scheme over the S-D, S-R and R-D links, in case of cooperative transmission mode, and over the S-D link in the case of only direct communication. For the sake of simplicity, it is also assumed that all wireless channels are subject to  Nakagami${-}m$ multipath fading conditions with $\Omega_{S,D} = \Omega_{S,R} = \Omega_{R,D} = 0$dB. The involved  path-loss effects are modeled by an example model of  $ P_{L_{i,j}}\, [{\rm dB}] = 148 + 40\log_{10}(d_{i,j} [{\rm km}])$, which is also used in device-to-device based communications \cite{SP}, and thus applies to mobile relaying as well.  Furthermore, in order to simplify the geometry-related calculations, we assume that all nodes are located along a straight line, which satisfies the distance  relationship  $ d_{S,D} = d_{S,R} + d_{R,D} $. However, it is recalled here that the path-loss and distance assumptions are only indicative in the context of the considered    examples, while the provided analysis and optimization frameworks are valid more generally. In this context, we further assume the following system parameters: $N_{0} = -174 \text{dBm/Hz}$; $\eta = 0.35$; $T_{tr} = 5{\rm \mu }s$; $L = 2\text{kbits}$, and  $P_{LO} = 50\text{mW}$ \cite{X:Goldsmith, Bahai}. We also use the constant circuit powers as $P_{CT_x} = 100\text{mW}$, $ P_{CR_x} = 150\text{mW}$  and the maximum transmission power $ P_{maxt} = 1000\text{mW}$. The  bandwidth of the system is assumed to be $B = 200\text{kHz}$ and the noise figure $ N_{f} = 6\text{dB} $. Due to the linearity requirement of the $M{-}$QAM signals, the value of the drain efficiency is assumed   $\eta = 0.35$, which is a practical value for class-A and AB RF power amplifiers (PA). The considered system parameters are depicted in Table  I and are used unless otherwise stated.

\begin{table}[htb!]\label{table:System}
\centering
\caption{Assumed System Parameters}
\begin{tabular}{|c|c|}
\hline\hline
$ N_{0} = \text{-174dBm/Hz}  $ & $  N_{f} = \text{6dB} $ \\
\hline
$ B = \text{200kHz}$ & $L = \text{2kbits}$ \\
\hline
$ P_{CT_x} = \text{100mW}  $ & $ P_{CR_x} = \text{150mW} $ \\
\hline
$ P_{LO}= \text{50mW} $& $ P_{maxt} = \text{1000mW} $ \\
\hline
$ \eta = 0.35 $ &$ T_{tr} = 5\mu s.$ \\
\hline
\end{tabular}
\end{table}
\indent 
 We commence by analyzing the minimum energy per information bit required  for the direct and cooperative transmissions when the relay node is taken into account and placed in different locations. The location of the relay node is  represented with  parameter $ f = d_{S,R}/d_{S,D} $. Fig. 5   illustrates the total energy consumption per information bit as a function of the transmission distance from source to destination  for $4-$QAM/QPSK with fading parameter of  $m =1.25$, destination  target BER  of $ 10^{-2} $ and zero spatial correlation under the maximum transmit power constraint. The transmit power allocation is carried out by the derived OPA scheme resulting to the indicative values  in Table II. It is observed that   distance thresholds separate the regions where DT performs better than CT and vice-versa. Furthermore, it is shown that when the relay is located in the middle, i.e. the source-relay distance equals the relay-destination distance ($ f = 0.5 $), renders the best energy efficiency among all relay locations.  This indicates that the configuration is almost symmetric in the source-relay and relay-destination distances, which assists the system to operate robustly in transmission over severe fading   conditions. However,  it is shown that at relatively small distances (here  $ 0 \leq d_{S,D}\leq 170 m $), the exact location of the relay does not affect substantially the performance of the cooperative system as it appears to remain almost the same in all considered scenarios. This renders the relay positioning and planning rather simple when the relay falls within this range, while it additionally provides insight e.g. for relay selection algorithms in the case of randomly distributed relays in emerging relay-based wireless networks.
\begin{figure}[tp!]
\centering{\includegraphics[keepaspectratio,width= 15cm]{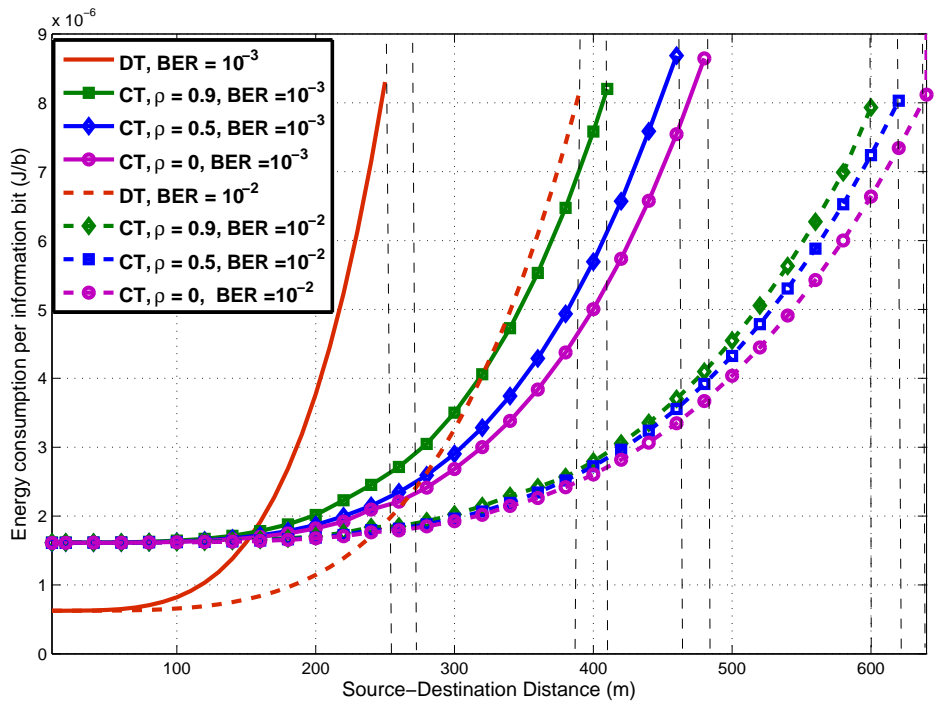} }
\caption{Energy consumption per information bit versus source-destination distance when the relay is located in the middle over spatially correlated Nakagami${-}1.25$ for $4{-}$QAM/QPSK constellations with different target BERs.}
\end{figure}

\begin{footnotesize}
\begin{table}[t]
\caption{ Optimal Transmit Power Values for Source and Relay for  Different Relay Locations over Uncorrelated  Nakagami${-}1.25$ Links at target BER of ${10^{-2}}$ for $4{-}$QAM / QPSK Modulation.}
\centering
\begin{tabular}{|c|c|c|c|c|c|c|}
\hline
\hline 
   &\multicolumn{2}{c|}{Relay close to $ S \;(f = 0.1)$} & \multicolumn{2}{c|} { Relay in middle \;$(f = 0.5 )$ } & \multicolumn{2}{c|} {Relay close to $D \;(f = 0.9)$} \\ 
\hline 
D(m) & $P_{S}(\text{W})$ & $ P_{R}(\text{W})$ & $ P_{S} (\text{W})$ & $ P_{R}(\text{W}) $ & $ P_{S} (\text{W})$ & $ P_{R} (\text{W})$ \\ 
\hline 
100 & 0.00033 & 0.0006  &0.000312 &0.000197& 0.000651 & 0.00049 \\ 
\hline 
200 & 0.016& 0.0127 & 0.0057 & 0.0033 & 0.0104 &0.0101 \\ 
\hline 
300 & 0.0877 & 0.0595 & 0.0289& 0.0223& 0.0527 & 0.0525 \\ 
\hline 
400 & 0.3366  &0.1537 & 0.0706 & 0.0703 & 0.1666 & 0.0198\\ 
\hline 
480 & 0.6225 & 0.3601 & 0.1466 & 0.1455 & 0.3455 & 0.0342 \\ 
\hline 
620 & - & -  & 0.4107  &0.402 & 0.9627& 0.0373 \\ 
\hline 
650 & -&- &0.4970 &0.483 & - & - \\
\hline
\end{tabular} 
\end{table}
\end{footnotesize}

\begin{figure}[tbp!]
\centering{\includegraphics[keepaspectratio,width= 15cm]{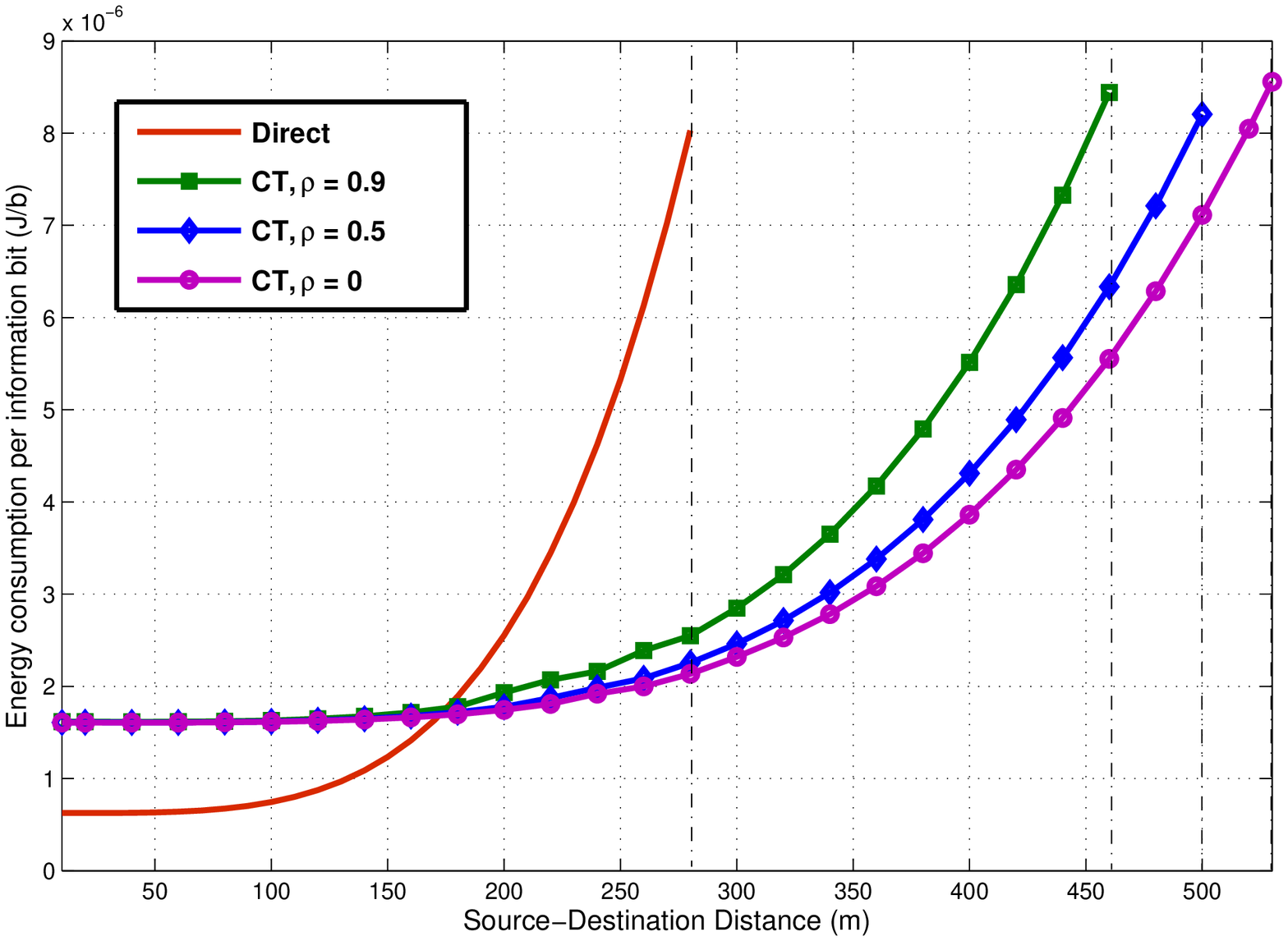} }
\caption{Energy consumption per information bit versus source-destination distance when the relay is  located in the middle over spatially  correlated Nakagami${-}0.75$ fading channels at target BER of  $ 10^{-2}$ for $4{-}$QAM/QPSK constellation and different spatial correlation values.} 
\end{figure}
The corresponding energy efficiency is also analyzed for  target BERs of $ 10^{-2} $ and  $ 10^{-3} $. Fig. 6                     
illustrates the average total energy consumption per information bit for $4-$QAM${/}$QPSK in both direct and cooperative transmissions in  Nakagami${-} 1.25$  fading conditions under the  maximum transmit power constraint and the following spatial correlation scenarios: $ \rho = \{0, 0.5, 0.9\}$. Also, the relay node is located in the middle and the transmit power is allocated optimally to the source and relay nodes in all cases. It is observed that for the fixed target BERs of $ 10^{-2} $ and $ 10^{-3}$, the direct scheme outperforms the cooperative transmission only at average source-destination distances below 240m and 150m, respectively. On the contrary, for average distances greater than  240m and 150m,  CT becomes more energy efficient as the transmit power constitutes a significant share of the average total energy consumption even under the worst spatial correlation scenario. Furthermore, it is shown that for the given target BERs the DT schemes attain maximum transmission distances of 390m and 250m, respectively,  under the given maximum transmission power constraint while in both cases  the cooperative transmission schemes extend to substantially longer distances. However, these advantages vary according to the level of the involved spatial correlation where the improvement in energy efficiency is inversely proportional to $\rho$,  in both scenarios. The reason    is that for every step of transmission distance,  greater proportion of power is assigned to the  source and relay nodes in order to overcome performance losses incurred by the spatial correlation, within the given resource constraints.  Concrete examples are shown in Table III and Table IV  for some indicative transmission distances and the two target BERs; there, the energy savings using CT scheme, defined as $ 1- \overline{E}^{C}/\overline{E}^{D} $ for $\rho = 0.9$, $\rho = 0.5$ and $\rho = 0$ for the target  BER of $ 10^{-2} $ at $d_{S,D} = 390$m are $ 65\% $, $ 67\%$ and $ 69.3\% $, respectively,  whereas for a target BER of $ 10^{-3} $ at $d_{S,D} = 250$m the energy reduction is $69\% ,72 \% \;\text{and}\; 73.8\%$, respectively. Interestingly, beyond a critical distance of 320m, even the highly spatially correlated CT mode at target BER of $10^{-3}$ exhibits better energy-efficiency than the DT scheme with target BER of  $10^{-2}$.

\begin{footnotesize}
\begin{table}[tbp!]
\caption{Optimal Transmit Power Values for Source and Relay for relay located in the middle over correlated Nakagami${-}1.25$ Links with target BER of ${10^{-2}}$ Using $4{-}$QAM ${/}$ QPSK Modulation.}
\centering
\begin{tabular}{|c|c|c|c|c|c|c|}
\hline
\hline 
Corr.  & \multicolumn{2}{c|}{$\rho = 0$} & \multicolumn{2}{c|} {$\rho = 0.5$} & \multicolumn{2}{c|} {$\rho = 0.9$ } \\ 
\hline 
D(m) & $P_{S} (\text{W})$ & $P_{R} (\text{W})$ & $P_{S} (\text{W}) $ &$ P_{R}(\text{W})$ & $P_{S} (\text{W}) $& $P_{R} (\text{W})$ \\ 
\hline 
100 & 0.000372 & 0.000197   & 0.000424   & 0.000233  & 0.000514  & 0.000296  \\ 
\hline 
200 & 0.0057 & 0.0033 & 0.0064 & 0.0039 & 0.0077& 0.0049 \\ 
\hline 
300 & 0.0225 & 0.0223 & 0.0249 & 0.0249 & 0.0290 & 0.0290 \\ 
\hline 
400 & 0.0706 & 0.0703 & 0.0788 & 0.0785 & 0.0918 & 0.0915 \\ 
\hline 
500 & 0.1728 & 0.1712 & 0.1929 & 0.1912 & 0.2249  & 0.2228 \\
\hline 
600 &0.3598 &0.3530 & 0.4021 & 0.3944 & 0.4692  & 0.4602 \\
\hline
630 & 0.4381 & 0.4281 & 0.4897 &0.4785 &-&- \\
\hline 
650 & 0.4970 & 0.4843 & - & - & - & - \\ 
\hline 
\end{tabular} 
\end{table}
\end{footnotesize}

In the same context, Fig. 7 illustrates the average total energy per information bit for both direct and cooperative transmissions for  $ m = 0.75 $, which corresponds to  severe fading conditions.  The target BER is set to $ 10^{-2} $ under the given transmit power constraint while $\rho = \{0, 0.5, 0.9\}$. The transmit power is again allocated optimally to the source and relay with the latter  positioned in the center of the network as shown in Table V.  It is shown  that DT outperforms CT only when $ d_{S,D} \leq 170m $. However, as the distance increases beyond this point, the corresponding overall  benefits by CT are significant even under the worst spatial correlation scenario.  Indicatively,  at a transmission distance of 280m,  which is the maximum distance that DT can operate with the available maximum transmission power, the energy gains by  CT   are    $ 68\% $, $ 72\% $  and  $ 74\%$, respectively.  In addition, it is shown that the advantage of the cooperation is more significant in severe fading conditions.

\begin{footnotesize}
\begin{table}[t!]
\caption{Optimal Transmit Power Values for Source when the  Relay is  Located in the Middle over Correlated  Nakagami${-} 1.25$ Links with Target BER of ${10^{-3}}$ Using $4{-}$QAM ${/}$ QPSK Constellations.}
\centering
\begin{tabular}{|c|c|c|c|c|c|c|}
\hline
\hline 
Corr.  & \multicolumn{2}{c|}{$\rho = 0$} & \multicolumn{2}{c|} {$\rho = 0.5$} & \multicolumn{2}{c|} {$\rho = 0.9$} \\ 
\hline 
D(m) & $P_{S} (\text{W})$ & $P_{R} (\text{W})$ & $ P_{S} (\text{W})$ & $ P_{R}(\text{W}) $ & $ P_{S} (\text{W})$ & $P_{R} (\text{W})$ \\ 
\hline 
100 & 0.00012  & 0.0007   & 0.00014   & 0.0009  & 0.00022  & 0.00014   \\ 
\hline 
200 & 0.0179 & 0.0116  & 0.0218 & 0.0147 & 0.0331 & 0.0233 \\ 
\hline 
300 &0.0762  & 0.0762 & 0.0914 &0.0914  & 0.1332 & 0.1332 \\ 
\hline 
410 & 0.2662 & 0.2651  & 0.3196 & 0.3183 & 0.4660 & 0.4641 \\ 
\hline 
440 & 0.3533 & 0.3513 & 0.4242 & 0.4218 & - & - \\
\hline 
480&0.5009 & 0.4969 &-&-&-&  -\\
\hline 
\end{tabular} 
\end{table}
\end{footnotesize}

\indent 
\begin{footnotesize}
\begin{table}[t!]
\caption{Optimal Transmit Power Values for Source and Relay when the Relay is Located in the Middle over Spatially Correlated  Nakagami${-}0.75$ Links with target BER of ${10^{-2}}$ Using $4{-}$QAM ${/}$ QPSK Modulations.}
\centering
\begin{tabular}{|c|c|c|c|c|c|c|}
\hline
\hline 
Corr.  & \multicolumn{2}{c|}{$\rho = 0$} & \multicolumn{2}{c|} {$\rho = 0.5$} & \multicolumn{2}{c|} {$\rho = 0.9 $} \\ 
\hline 
D(m) & $P_{S} (\text{W})$ & $P_{R} (\text{W})$ & $P_{S} (\text{W})$ & $P_{R}(\text{W})$ & $P_{S} (\text{W})$ &$ P_{R} (\text{W})$ \\ 
\hline 
100 & 0.000754  & 0.0004635   & 0.0009112   & 0.0005811  & 0.00013  & 0.0009   \\ 
\hline 
200 & 0.0119 &0.0076  & 0.0143 & 0.0093  &0.0386 & 0.0097  \\ 
\hline 
300 & 0.0509 &0.0509 & 0.0608 & 0.0608&0.0873  & 0.0873  \\ 
\hline
400 &0.1612 & 0.1606 & 0.1926 &0.1918  & 0.2765 & 0.2753  \\ 
\hline
460 & 0.2823 & 0.2802  & 0.3373  & 0.3348 & 0.4845 & 0.4808 \\
\hline 
500 & 0.3945 & 0.3905 &  0.4714 &0.4666 &-&  -\\
\hline 
530 & 0.4985 & 0.4921 &  - & -  &-&  -\\
\hline
\end{tabular} 
 \end{table} 
 \end{footnotesize}
\begin{figure}[tbp!]
\centering{\includegraphics[keepaspectratio,width= 15cm]{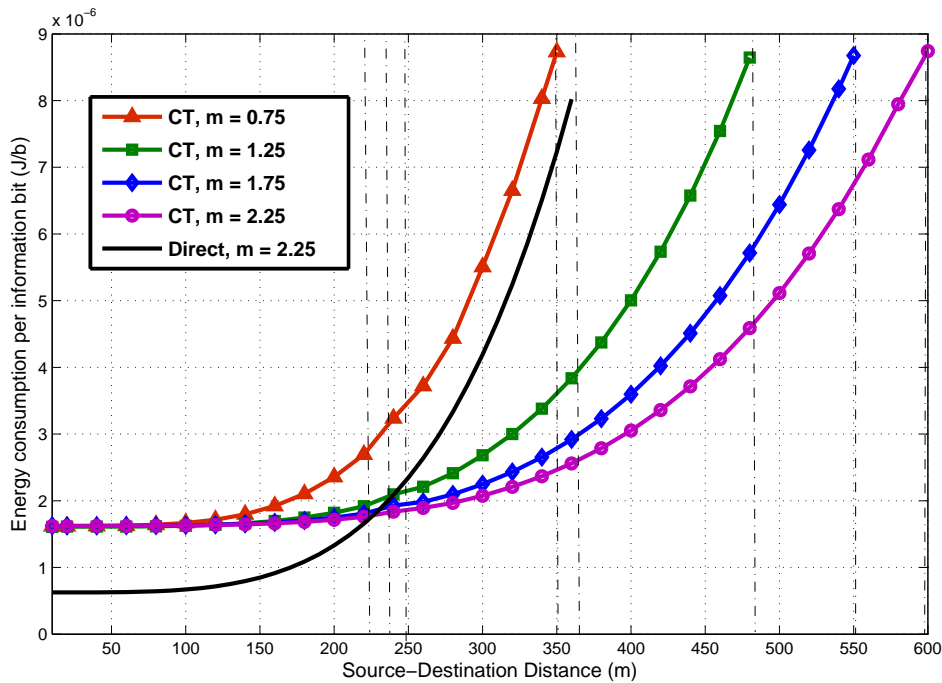} }
\caption{Energy consumption per information bit versus source-destination distance when the relay is located in the middle over uncorrelated Nakagami${-}m$ fading conditions at target BER of  $ 10^{-3}$ for $4{-}$QAM ${/}$ QPSK constellation.} 
\end{figure}

\indent 
Fig. 8 illustrates the average total energy consumption per information bit required for CT and DT  as a function of S-D distance for fading parameters, $ m = \{0.75, 1.25,1.75,  2.25\}$ for CT  and $ m = 2.25$  for  DT. The target BER is set to  $10^{-3}$, the transmit power is allocated optimally, the relay is located in the middle and a zero spatial correlation is assumed. It is observed that the critical distances below which DT  outperforms the CT in terms of energy efficiency are  250m, 230m and 220m for $ m = 1.25, m = 1.75, \text{and} \; m = 2.25 $  , respectively. Moreover, the analysis indicates  that DT with non-severe multipath fading condition (Nakagami${-2.25}$) can operate  only up to 360m before utilizing the maximum transmit power, while the CT  extends significantly even for moderate fading conditions, except for the worst case scenario ($ m = 0.75 $). It is also shown that the gain from the cooperation is not uniform as the Nakagami parameter increases from $ m = 0.75 $ to $ m = 1.25 $,  from $ m = 1.25 $ to $ m = 1.75 $ and then from $ m = 1.75  $  to $ m = 2.25$.

Finally, Fig. 9 depicts the cooperation gain, defined in \eqref{ss}, when the relay is located in the middle of the source and destination. The power allocation is again carried out by using the derived OPA scheme for target BER of $ 10^{-2}$ under the maximum transmission power constraint with $\rho = \{0, 0.5, 0.9 \} $ for $ m = 1.25 $. The transmission distance is limited to 390m since beyond this limit it is only the CT mode that can transmit    until its maximum transmission distance, depending on  the spatial correlation between S-D and R-D paths.  When the cooperation gain is below unity, the DT is actually more energy efficient than CT. As already mentioned, the reason behind this  is that when $CG\leq1$, which corresponds to relatively small transmission distances, the actual transmit power  constitutes only a small fraction of the total average power consumption. However,  when $ CG >  1 $, the system benefits significantly from   cooperation and in  general CG increases proportionally as the transmission distance increases for all scenarios of spatial correlation between the S-D and R-D paths. Interestingly, the existence of such efficiency threshold distance also implies that a hybrid system, where cooperation is only sought and deployed beyond certain minimum distance, can provide the most comprehensive solution to the energy-efficiency optimization. The analysis and modeling results and tools provided in this article form directly the basis for further development of such schemes in different communication scenarios, which forms an important topic of future work.

\begin{figure}[!btp!]
\centering{\includegraphics[keepaspectratio,width= 15cm]{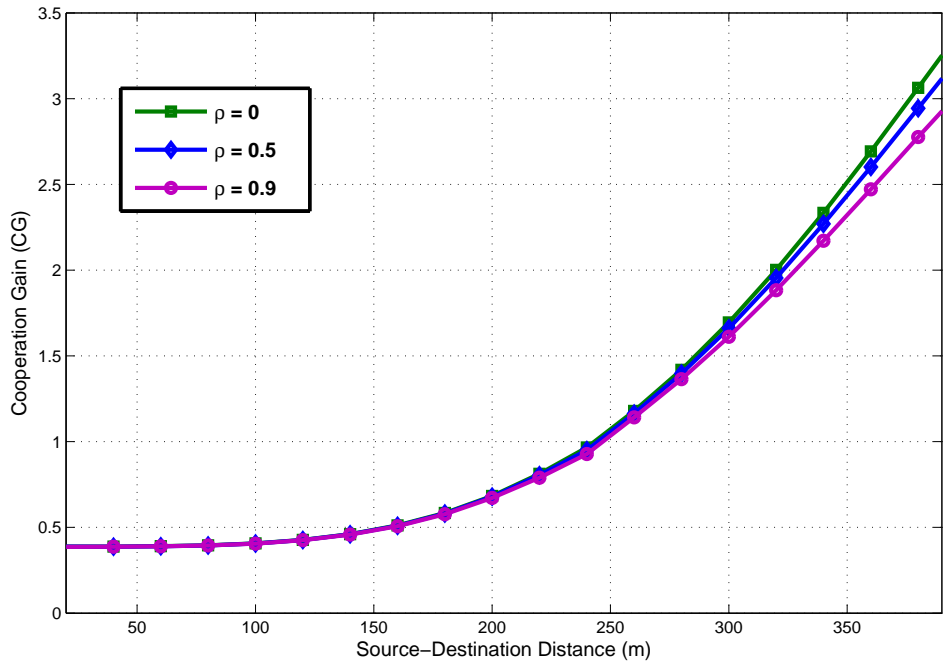} }
\caption{Cooperation gain versus source-destination distance when the relay is located in the middle over spatially correlated Nakagami$ {-}1.25$ fading environment at target BER of  $ 10^{-2}$ for $4{-}$QAM ${/}$ QPSK  constellations. After source-destination distance of 390m, direct transmission runs out of power and cannot anymore reach the target BER.}
\end{figure}

\section{Conclusions}

This work was devoted to the end-to-end SER analysis as well as  the energy efficiency analysis and optimization of both direct and regenerative  cooperative transmissions over Nakagami${- m}$ fading conditions in the presence of spatial correlation.  Novel  closed-form expressions were firstly derived for the symbol-error rate of both $M-$QAM and $M-$PSK constellations which were  subsequently employed in formulating the constrained energy analysis and optimization problems under destination bit-error-rate target  and  maximum transmit power constraints considering both transmit energy as well as the energy consumed by the transceiver circuits. The corresponding   results indicate that depending on the severity of multipath fading, spatial correlation between the source-destination and relay-destination paths and the location of the relay node, the direct transmission can be  more energy efficient than cooperative transmission but only for rather short transmission distances and up to  a certain threshold value. Beyond this value, the system, as expected,  benefits substantially from relaying and the corresponding  cooperation gain increases proportionally to the  transmission  distance. It is expected that the offered results can  be useful in the design, dimensioning  and deployment of low-cost and energy efficient cooperative communication systems in the future, especially towards the green communications era where the requirements and incentives towards energy consumption optimization are considered critical.

\section*{Appendix }

\subsection{Proof of Theorem $1$}

The average SER for the direct transmission can be expressed as,

 \begin{equation} \label{L9}
\overline{\rm SER}^{D}_{D}  = F_{QAM} \left[ \left( 1 + \frac{(P_{S}/\: P_{L_{S,D}})\: \Omega_{S,D}\: g_{QAM}}{ N_{0}\: m_{S,D} \,{\sin}^{2}(\theta)}\right)^{-m_{S,D}} \right].   
\end{equation}
Evidently, a closed-form expression for   \eqref{L9} is subject to analytic evaluation of the following  integrals

\begin{equation} \label{L10}
\mathcal{I} \left(a, m; 0, \frac{\pi}{2} \right) = \int_{0}^{\pi {/}2} \frac{1}{ \left( 1 + \frac{a}{{\sin}^{2}(\theta)}\right)^{m} } d \theta 
\end{equation}
and

\begin{equation} \label{L11}
\mathcal{I} \left(a, m; 0, \frac{\pi}{4} \right) = \int_{0}^{\pi {/}4} \frac{1}{ \left( 1 + \frac{a}{{\sin}^{2}(\theta)}\right)^{m} } {\rm d} \theta.  
\end{equation}
By re-writing the indefinite form of the above class of integrals as 

\begin{equation}
\mathcal{I}(a,m) =   \int  \frac{\sin^{2m}(\theta)}{ \left( \sin^{2}(\theta) +  a \right)^{m} } {\rm  d} \theta
\end{equation}
and setting $u = \cos^{2}(\theta)$, one obtains

\begin{equation}
\mathcal{I}(a,m) = - \int \frac{(1-u)^{m - \frac{1}{2}}}{2 \sqrt{u} (1 - u + a)^{m} } {\rm d}u. 
\end{equation}
The above integral can be expressed in closed-form in terms of the Appell hypergeometric function of the first kind, namely

\begin{equation} \label{trigonometric_a}
\mathcal{I}(a,m) =    - \frac{\cos (\theta)}{(1 + a)^m} F_{1}\left( \frac{1}{2}, \frac{1}{2} - m, m, \frac{3}{2}, \cos^{2}(\theta), \frac{\cos^{2}(\theta)}{1+a} \right). 
\end{equation}
Equation \eqref{trigonometric_a} reduces to zero when $\theta = \pi {/}2$. To this effect, it immediately follows that  

\begin{equation} \label{trigonometric_b}
\mathcal{I}\left(a, m, 0, \frac{\pi}{2}\right) = \frac{1}{(1+a)^{m}} F_{1}\left( \frac{1}{2}, \frac{1}{2} - m, m, \frac{3}{2}, 1, \frac{1}{1+a} \right)
\end{equation}
which with the aid of the properties of $F_{1}(.)$ function can be equivalently expressed as 

\begin{equation} \label{trigonometric_c}
\mathcal{I}\left(a, m, 0, \frac{\pi}{2}\right) =  \frac{\sqrt{\pi} \Gamma\left(m + \frac{1}{2} \right)}{2(1+a)^{m} \Gamma(m + 1)} \,_{2}F_{1}\left( \frac{1}{2}, m, m + 1, \frac{1}{1+a} \right).  
\end{equation}
In the same context, 

\begin{equation} \label{trigonometric_d}
\mathcal{I}\left(a, m, 0, \frac{\pi}{4}\right) =  \frac{1}{(1+a)^{2}} \left\lbrace F_{1}\left(\frac{1}{2}, \frac{1}{2} - m, m, \frac{3}{2}, 1, \frac{1}{1+a} \right) - \frac{1}{\sqrt{2}}  F_{1}\left(\frac{1}{2}, \frac{1}{2} - m, m, \frac{3}{2}, \frac{1}{2}, \frac{1}{2(1+a)} \right)  \right\rbrace
\end{equation}
which can be alternatively expressed as follows: 

\begin{equation} \label{trigonometric_e}
\mathcal{I}\left(a, m, 0, \frac{\pi}{4}\right) =    \frac{\sqrt{\pi} \Gamma\left(m + \frac{1}{2} \right) \, _{2}F_{1}\left( \frac{1}{2}, m,  m+1, \frac{1}{1+a} \right)}{2(1+a)^{m} \Gamma(m+1)}  - \frac{F_{1}\left(\frac{1}{2}, \frac{1}{2} - m, m, \frac{3}{2}, \frac{1}{2}, \frac{1}{2(1+a)} \right)}{\sqrt{2}(1+a)^{m}}.  
\end{equation}
By performing the necessary change of variables in \eqref{trigonometric_c} and \eqref{trigonometric_e} and substituting  in \eqref{L9} yields \eqref{D_new}.

\subsection{Proof of Lemma 1}
The $\mathcal{J}(a, b, m)$ integral can be   re-written as 

\begin{equation}
\mathcal{J}(a, b, m) =\int  \frac{\sin^{4m}(\theta)}{(\sin^{4}(\theta) + a\sin^{2}(\theta) + b)^{m}} {\rm d}\theta. 
\end{equation}
By setting $u = \cos^{2}(\theta) $ it follows that

\begin{equation}
\mathcal{J}(a, b, m) = - \frac{1}{2}  \int  \frac{(1-u)^{2m}}{\sqrt{u} \sqrt{1 - u} \, [(1 - u)^{2} + a(1-u) + b]^{m}}  {\rm d} u. 
\end{equation}
By applying  the binomial theorem in \cite[eq. (1.111)]{B:Tables} one obtains

\begin{equation}
\mathcal{J}(a, b, m) = - \sum_{l = 0}^{2m - \frac{1}{2}} \binom{2m - \frac{1}{2}}{l} \frac{(-1)^{l}}{2} \int \frac{u^{l - \frac{1}{2}}}{[1 - u(a+2) + u^{2} + a + b]^{m}} {\rm d}u. 
\end{equation}
The above integral can be expressed in terms of the Appell function of the first kind yielding 

\begin{equation}
\begin{split}
\mathcal{J}(a, b, m) &= - \sum_{l = 0}^{2m - \frac{1}{2}} \binom{2m - \frac{1}{2}}{l} \frac{(-1)^{l}  u^{l + \frac{1}{2}} \, F_{1}\left(l + \frac{1}{2}, m, m, l + \frac{3}{2}, \frac{2u}{2+a-\sqrt{a^{2} - 4b}}, \frac{2u}{2+a + \sqrt{a^{2} - 4b}} \right) }{(1+2l) (a + b + (u - 1)^{2} - au)^{m}}   \\
& \qquad  \times \left(1 - \frac{2u}{2 + a - \sqrt{a^{2} - 4b}} \right)^{m}  \left(1 - \frac{2u}{2 + a - \sqrt{a^{2} - 4b}} \right)^{m} 
\end{split}
\end{equation}
which upon performing the necessary counter-substitution and  algebraic manipulations yields  \eqref{Lemma_1}. 

\subsection{Proof of Lemma 2}

The $\mathcal{K}(a, b, m, n)$ integral can be alternatively re-written as

\begin{equation}
\mathcal{K}(a, b, m, n) = \int \frac{\sin^{2m + 2n}(\theta)}{(\sin^{2}(\theta) + a)^{m}  \, (\sin^{2}(\theta) + b)^{n}}  {\rm d}\theta
\end{equation}
which by  setting $u = \cos^{2}(\theta)$ can be expressed as 

\begin{equation}
\mathcal{K}(a, b, m, n) = \int \frac{(1-u)^{m + n - \frac{1}{2}}}{2(1 - u + a)^{m} (1 - u + b)^{n} \sqrt{u}} {\rm d}u. 
\end{equation}
By applying the binomial theorem in \cite[eq. (1. 111)]{B:Tables}, it immediately follows that 

\begin{equation}
\mathcal{K}(a, b, m, n) = \sum_{l = 0}^{m + n - \frac{1}{2}} \binom{m + n - \frac{1}{2}}{l} \frac{(-1)^{l}}{2} \int \frac{u^{l - \frac{1}{2}}}{(1 - u 
+ a)^{m} (1 - u + b)^{n}} {\rm d}u.  
\end{equation}
The above integral can be expressed in closed-form in terms of the Appell function of the first kind. As a result, by making the necessary counter-substitution and   performing some long but basic algebraic manipulations, equation \eqref{new_9} is deduced.  

\subsection{Proof of Lemma 3}

By setting $u = \sin^{2}(\theta)$, it follows that 

\begin{equation}
\int \sin^{2m}(\theta) {\rm d}\theta = \int \frac{u^{m - \frac{1}{2}}}{2 \sqrt{1-u}} {\rm d}u. 
\end{equation}
The above integral can be expressed in closed-form in terms of the Gaussian hypergeometric function, namely, 

\begin{equation}
\int \frac{u^{m - \frac{1}{2}}}{2 \sqrt{1-u}} {\rm d}u = -\sqrt{1-u}\, _{2}F_{1}\left( \frac{1}{2},\frac{1}{2}  - m, \frac{3}{2}, 1 + u \right). 
\end{equation}
Therefore, by performing the counter-substitution, equation \eqref{sin_integral} is deduced.

\subsection{Proof of Convexity of the Optimization Problem }

Below we prove the existence of optimal powers, which are subsequently employed in minimizing  the overall energy consumption in the considered cooperative communication system. Based on \eqref{L41},  the average total power consumption can be re-written as  follows:

\begin{equation} \label{51}
\overline{P}^{C}_{T}  = \left( C_{1} + (1 + \alpha)P_{S}\right)           
+ \left(C_{2}+ (1 + \alpha) {P}_{R}\right) \left(1 - \overline{\rm SER}_{S,R} \right) 
\end{equation} 
where $ C_{1} = P_{CT_x} + 2P_{CR_x} \text{and}\;  C_{2} = P_{CT_x} + P_{CR_x}. $ The symbol-error-rate ${  \overline{\rm SER}_{S,R}}$ can be  expressed in closed-form with the aid of Theorem 1. This expression is a function of $1 + P_{S}\: \: \Omega_{S,R}\: g_{QAM}/P_{L_{S,R}}N_{0}\: m_{S,R}$; therefore, for proving the existence of the optimum values, it is sufficient to show that $ \partial^{2}\overline{P}^{C}_{T}/\partial^{2} {P}^{2}_{R} \geq 0 $ and $ \partial^{2}\overline{P}^{C}_{T}/\partial^{2}P^{2}_{S} \geq \:0 $. To this end, it is straightforward to  show that $ \partial^{2}\overline{P}^{C}_{T}/\partial^{2} {P}^{2}_{R} = 0 $. Likewise, based on the   optimal condition in \eqref{L59} and taking the second-order partial derivative w.r.t the variables $P_{S}$ and $ {P}_{R}$, one obtains  

\begin{equation}\label{52}
\frac{\partial^{2}\overline{P}^{C}_{T}}{\partial^{2}P^{2}_{S}} \geq \frac{\partial^{2}\overline{P}^{C}_{T}}{\partial^{2}P_{S} {P}_{R}}
\end{equation}
where 

\begin{equation}\label{52}
\frac{\partial^{2}\overline{P}^{C}_{T}}{\partial^{}P_{S}P_{R}} =  \frac{(1+\alpha_1 ) m_{S, R}\Omega_{S,R}g_{QAM}}{(1+a_1)N_{0}P_{L_{S,R}}}K_{4} 
\end{equation} 
 and

\begin{equation}
K_{4} = \frac{4 C}{\pi} \mathcal{I}  \left(\frac{P_{S}  \Omega_{S,R}\: g_{QAM}}{ N_{0}  P_{L_{S,R}} \: m_{S,R} }, m_{S, R}; \frac{\pi}{2} \right)  - \frac{4C^{2}}{\pi}  \mathcal{I}  \left(\frac{P_{S} \Omega_{S,R}\: g_{QAM}}{ N_{0}\: P_{L_{S,R}} m_{S,R}  }, m_{S, R}; \frac{\pi}{4} \right)  
\end{equation}
where it is recalled that $C = 1 - 1/\sqrt{M}$ for the case of $M-$QAM whereas $K_{4}$ denotes  the SER representation with values in the range $0\leq {\rm SER} \leq 1$ and  with all other constants being positive. Based on this,  the second-order partial derivatives w.r.t $P_{S}$ and $ {P}_{R}$ are always greater  than or equal to zero, which implies that   $\partial^{2}\overline{P}^{C}_{T}/\partial^{2}P^{2}_{S} \geq 0$. Given the general second order conditions in \cite{body}, it immediately follows that \eqref{51} is convex w.r.t to $P_{S}$ and ${P}_{R}$ and possesses a unique minimum value.

\bibliographystyle{IEEEtran}

\emph {IEEE Commun. Surveys and Tuts.}, vol. 15, no. 1, pp. 167${-}$178, 1st Quart. 2013. 

\bibitem {SEE}
T. Ma, M. Hempel, D. Peng, and H. Sharif,
 ``A survey of energy-efficient compression and communication techniques for multimedia in resource constrained systems,"
 \emph{IEEE Commun. Sur. and Tuts.}, vol. 15, no. 3, pp. 963${-}$972, $3^{\rm rd}$ Quart., 2013.

 \bibitem{Final_6} 
 P. C. Sofotasios, T. A. Tsiftsis, Yu. A. Brychkov, S. Freear, M. Valkama, and G. K. Karagiannidis, 
 ``Analytic Expressions and Bounds for Special Functions and Applications in Communication Theory,"
 \emph{ IEEE Trans. Inf. Theory}, vol. 60, no. 12, pp. 7798$-$7823, Dec. 2014.

\bibitem{Final_1}
G. K. Karagiannidis,
``On the symbol error probability of general order rectangular QAM in Nakagami$-m$ fading,''
 \emph{IEEE Commun. Lett.}, vol. 10, no. 11, pp. 745${-}$747, Oct. 2006.
 
 \bibitem{Final_7} 
K. Ho-Van, P. C. Sofotasios, 
``Outage Behaviour of Cooperative Underlay Cognitive Networks with Inaccurate Channel Estimation,"
\emph{ in Proc. IEEE ICUFN '13}, pp. 501${-}$505, Da Nang, Vietnam, July 2013.

\bibitem{C} 
Y-W. Hong, W-J. Huang, F-H. Chiu, and C-C. J. Kuo, ``Cooperative communications in resource-constrained wireless networks,"
 \emph{IEEE Signal Process. Mag.}, vol. 24, no. 3, pp. 47${-}$57, May 2007.
 
 \bibitem{Final_8} 
 K. Ho-Van, P. C. Sofotasios, 
 ``Bit Error Rate of Underlay Multi-hop Cognitive Networks in the Presence of Multipath Fading,"
 \emph{ in IEEE ICUFN '13}, pp. 620${-}$624, Da Nang, Vietnam, July 2013. 
 
\bibitem{Final_2}
D. A. Zogas, G. K. Karagiannidis, and S. A. Kotsopoulos,
``Equal gain combining over Nakagami$-n$ (Rice) and Nakagami$-q$ (Hoyt) generalized fading channels,'' 
\emph{IEEE Trans. Wireless Commun.}, vol. 4, no. 2, pp. 374$-$379, Apr. 2005. 

  \bibitem{Final_9} 
K. Ho-Van, P. C. Sofotasios, 
``Exact BER Analysis of Underlay Decode-and-Forward Multi-hop Cognitive Networks with Estimation Errors,"
\emph{IET Communications}, vol. 7, no. 18, pp. 2122${-}$2132, Dec. 2013.

\bibitem{S} 
A. Sendonaris, E. Erkip, and B. Aazhang, ``User cooperation diversity part I: System description,"
 \emph{IEEE Trans. Commun.}, vol. 51, no. 11, pp. 1927${-}$1938, Nov. 2003.

\bibitem{Final_10} 
K. Ho-Van, P. C. Sofotasios, S. Freear, 
``Underlay Cooperative Cognitive Networks, with Imperfect Nakagami${-}m$ Fading Channel Information and Strict Transmit Power Constraint,"
\emph{IEEE KICS Journal of Communications and Networks}, vol. 16. no. 1, pp. 10${-}$17, Feb. 2014. 
 
\bibitem{Final_3}
D. S. Michalopoulos, and G. K. Karagiannidis,
``Distributed switch and stay combining (DSSC) with a single decode and forward relay,''
\emph{IEEE Commun. Lett.}, vol. 11, no. 5, pp. 408$-$410, Nov. 2007.

\bibitem{Final_11}
M. K. Fikadu, P. C. Sofotasios, M. Valkama, and Q. Cui, 
``Analytic performance evaluation of $M-$QAM based decode-and-forward relay networks over enriched multipath fading channels,'' 
\emph{in IEEE WiMob '14}, Larnaca, Cyprus,  Oct. 2014, pp. 194$-$199.

\bibitem{Final_4}
D. S Michalopoulos, G. K Karagiannidis, T. A. Tsiftsis, and R. K. Mallik,
``Wlc41-1: An optimized user selection method for cooperative diversity systems,''
\emph{IEEE GLOBECOM `06}, San Fransisco, CA, USA, pp. $1-6$.

\bibitem{J} 
J. N. Laneman, D. N. C. Tse, and G. W. Wornell, 
``Cooperative diversity in wireless networks: Efficient protocols and outage behavior,"
\emph{IEEE Trans. Inf. Theory}, vol. 50, no. 12, pp. 3062${-}$3080, Dec. 2004.

\bibitem{Final_5}
D. S. Michalopoulos, A. S. Lioumpas, G. K. Karagiannidis, and R. Schober,
``Selective cooperative relaying over time-varying channels,''
\emph{IEEE Trans. Commun.}, vol. 58, no. 8, pp. 2402$-$2412, Aug. 2008.

\bibitem{Bahai}
S. Cui, A. Goldsmith, and A. Bahai, 
``Energy-constrained modulation optimization for coded systems," \emph{in Proc. IEEE Globecom 2003}, Dec. 2003,  pp. 372${-}$376.

\bibitem{R} 
R. Devarajan, S. C. Jha, U. Phuyal, and V. K. Bhargava,
 ``Energy-aware user selection and power allocation for cooperative communication system with guaranteed quality-of-service,"
\emph{in Proc. IEEE  Canadaian Work shop on Inf. Theory (CWIT)}, May 2011, pp. 216${-}$220.  

\bibitem{Z} 
Z. Zhou, S. Zhou, J.-H. Cui, and S. Cui, 
``Energy-efficient cooperative communication based on power control and selective single-relay in wireless sensor networks,"
 \emph{IEEE Trans. on Wireless Commun.}, vol. 7, no. 8, pp. 3066${-}$3078, Aug. 2008.

 \bibitem{W} 
C. Schurgers, O. Aberthorne, and M. B. Srivastava, ``Modulation scaling for energy aware communication systems,"
\emph{in Proc. IEEE  International Conference on Low Power Electronics and  Design}, Aug. 2001, pp. 96${-}$99.

\bibitem{X:Goldsmith}
S. Cui, A. Goldsmith, and A. Bahai,
``Modulation optimization under energy constraints,"
 \emph{in Proc. IEEE International Conference on Commun.}, May 2003, vol. 4, pp. 2805${-}$2811.


\bibitem{QAM2}
F. S. Al-Qahtani, T. Q. Duong, A. K. Gurung, and V. N. Q. Bao, 
``Selection decode-and-forward relay networks with rectangular QAM in Nakagami${-}m$ fading channels,"
\emph{in Proc. IEEE Wireless Comm. and Netw. Conf.}, Apr. 2010, pp. 1${-}$4.

 \bibitem{RQ}
 Q. Chen and  M.C. Gursoy,
 ``Energy efficiency analysis in amplify-and-forward and decode-and-forward cooperative networks,"
 \emph{in Proc. IEEE  Wireless Commun. and Netw. Conference (WCNC)}, Apr. 2010, pp. 1${-}$6.    

\bibitem{QAM}
S. S. Ikki, O. Amin, and M. Uysal, 
``Performance analysis of adaptive $L{-}$QAM for opportunistic decode-and-forward relaying,"
\emph{in Proc. IEEE  Vehicular Technol. Conference (VTC 2010)}, May 2010, pp. 1${-}$5.


 \bibitem{A} 
S. Cui, A. J. Goldsmith, and A. Bahai, ``Energy-efficiency of MIMO and cooperative MIMO techniques in sensor networks,"
\emph{IEEE J. Select. Areas Commun.}, vol. 22, no. 6, pp. 1089${-}$1098, Aug. 2004.

\bibitem{EES}
C-C. Kao, J. Wu,and S-C. Chen,
``Energy efficient clustering communication protocol for wireless sensor network," \emph{in Proc. IEEE Advanced Commun. Technol. Conference}, Feb. 2010,  pp. 830${-}$833.

\bibitem{KM}  
M.T. Kakitani, G. Brante, S.R. Demo, and A. Munaretto, 
``Comparing the energy efficiency of single-hop,multi-hop and incremental decode-and-forward in multi-relay wireless sensor networks," \emph{In Proc. IEEE International Symposium on Personal Indoor and Mobile Radio Commun.}, Sept. 2011, pp. 970${-}$974.

 \bibitem{K} 
M.T. Kakitani,S.R. Demo, and M.A. Imran, ``Energy efficiency contours for amplify-and-forward and decode-and-forward
cooperative protocols," \emph{In Proc. IEEE Int. Symp. on Commun. Systems, Netw. and Dig.  Signal Proc.}, July 2012, pp. 1${-}$5.

\bibitem{EE_3}
Z. Sheng, B. J. Ko, and K. K. Leung,
``Power Efficient decode-and-forward cooperative relaying,"
\emph{IEEE Wirelless Commun. Lett.}, vol. 1, no. 5, pp. 444${-}$447, Oct. 2012.

\bibitem{GL}
G. Lim and L.J. Cimini,
``Energy-efficient cooperative beamforming in clustered wireless networks," \emph{IEEE Trans. Wireless Commun.}, vol. 12, no. 3, pp. 1376${-}$1385,  March 2013.


\bibitem{YZ}
Y. Zhou, H. Liu, Z. Pan, L. Tian, J. Shi, and G. Yang,  ``Two-stage cooperative multicast transmission with
optimized power consumption and guaranteed coverage," \emph {IEEE J. Sel. Areas Commun.}, vol. 32, no. 2, pp. 274{$-$}284, Feb. 2014.

\bibitem{YX}
Y. Xu1, Z. Bai, B.Wang, P. Gong, and K. Kwak, ``Energy-efficient power allocation scheme for multirelay cooperative communications,"
\emph{in Proc. IEEE Advanced Commun. Technol. Conference}, Feb. 2014, pp. 260${-}$264. 

\bibitem{WW} 
W. Ji and B. Zheng,
``Energy efficiency based cooperative communication in wireless sensor networks," \emph{in Proc. IEEE International conference on Commun. Technol.}, Nov. 2010, pp. 938${-}$941.

\bibitem{WE}                                               
W. Fang, F. Liu, F. Yang, L. Shu, and S. Nishio,
``Energy-efficient cooperative communication for data transmission in wireless sensor networks,"
 \emph{IEEE Trans. Consum. Electron.},   vol. 56, no. 4, pp. 2185${-}$2192, Nov. 2010.  

 \bibitem{GG}
E. Kurniawan, S. Rinit, and  A.  Goldsmith,
``Energy  efficient  cooperation  for two-hop  relay  networks,"
 \emph{ in Proc. IEEE  Signal and Inf. Process. Association Annual Summit and Conference}, Dec. 2012, pp. 1${-}$10.

\bibitem{IV}
T.-D. Nguyen, O. Berder, and O. Sentieys,
``Energy-efficient cooperative techniques for infrastructure-to-vehicle communications," \emph{IEEE Trans. Intell. Transp. Syst.}, vol. 12, no. 3, pp. 659${-}$668, Sep. 2011.

 \bibitem{Sofotasios_1} 
P. C. Sofotasios,
\emph{On Special Functions and Composite Statistical Distributions and Their Applications in Digital Communications over Fading Channels}, Ph.D. Dissertation, University of Leeds, England, UK, 2010. 

\bibitem{Sofotasios_2}
S, Harput, P. C. Sofotasios, and S. Freear, 
``A Novel Composite Statistical Model For Ultrasound Applications," 
\emph{Proc. IEEE IUS `11}, pp. 1${-}$4, Orlando, FL, USA, 8${-}$10 Oct. 2011.  

\bibitem{Costa_3}
D. B. da Costa and S. Aissa, 
``Capacity analysis of cooperative systems with relay selection in Nakagami${-}m$ fading,"
 \emph{IEEE Commun. Lett.}, vol. 13, no. 9, pp. 637${-}$639, Sep. 2009. 

\bibitem{Sofotasios_3}
P. C. Sofotasios, T. A. Tsiftsis, K. Ho-Van, S. Freear, L. R. Wilhelmsson, and M. Valkama, 
``The $\kappa-\mu$/inverse-Gaussian composite statistical distribution in RF and FSO wireless channels,''
\emph{in IEEE VTC '13 - Fall}, Las Vegas, USA,  Sep. 2013, pp. 1$-$5.

 \bibitem{Trung} 
T. Q. Duong, V. N.Q. Bao, and H. J. Zepernick, 
``On the performance of selection decode-and-forward relay networks over Nakagami${-}m$ fading channels,"
\emph{IEEE Commun. Lett.},
vol. 13, no. 3, pp. 172${-}$174, Mar. 2009.

\bibitem{Sofotasios_4}
P. C. Sofotasios, T. A. Tsiftsis, M. Ghogho, L. R. Wilhelmsson and M. Valkama, 
``The $\eta-\mu$/inverse-Gaussian Distribution: A novel physical multipath/shadowing fading model,''
\emph{in IEEE ICC '13}, Budapest, Hungary, June 2013. 

 \bibitem{New_1}
 T. Q. Duong, G. C. Alexandropoulos, T. A. Tsiftsis, and H. Zepernick,
  ``Outage probability of MIMO AF relay networks over Nakagami${-}m$ fading channels," 
  \emph{Electronics Lett.,} vol. 46, no. 17, pp. 1229${-}$1231, Aug. 2010.

\bibitem{Sofotasios_5}
P. C. Sofotasios, and S. Freear, 
``The $\alpha-\kappa-\mu$/gamma composite distribution: A generalized non-linear multipath/shadowing fading model,''
\emph{IEEE INDICON  `11}, Hyderabad, India, Dec. 2011. 

\bibitem{New_4}
 T. Q. Duong, G. C. Alexandropoulos, H. Zepernick, and T. A. Tsiftsis, 
 ``Orthogonal space-time block codes with CSI-assisted amplify-and-forward relaying in correlated Nakagami${-}m$ fading channels," 
 \emph{IEEE Trans.  Veh. Technol.}, vol. 60, no. 3, pp. 882${-}$889, March. 2011. 

\bibitem{Sofotasios_10}
P. C. Sofotasios, and S. Freear,
``The $\kappa-\mu$/gamma composite fading model,''
\emph{IEEE ICWITS  `10}, Honolulu, HI, USA, Aug. 2010, pp. 1$-$4.

\bibitem{Sofotasios_6}
P. C. Sofotasios, and S. Freear,
``The $\alpha-\kappa-\mu$ extreme distribution: characterizing non linear severe fading conditions,'' 
\emph{ATNAC `11}, Melbourne, Australia, Nov. 2011. 

\bibitem{Shi}
Q. Shi and Y. Karasawa, 
``Error probability of opportunistic decode-and-forward relaying in Nakagami${-}m$ fading channels with arbitrary $m$," 
\emph{IEEE Wireless Commun. Lett.}, vol. 2, no. 1, pp.  86${-}$89, Feb. 2013.

\bibitem{Sofotasios_7}
P. C. Sofotasios, and S. Freear, 
``The $\eta-\mu$/gamma and the $\lambda-\mu$/gamma multipath/shadowing distributions,'' 
\emph{ATNAC  `11}, Melbourne, Australia, Nov. 2011. 

\bibitem{Sofotasios_8}
P. C. Sofotasios, and S. Freear, 
``On the $\kappa-\mu$/gamma composite distribution: A generalized multipath/shadowing fading model,'' 
\emph{IEEE IMOC `11},  Natal, Brazil, Oct. 2011, pp. 390$-$394.

\bibitem{Add_3}
S. S. Ikki and  M. H. Ahmed,
``Multi-branch decode-and-forward cooperative diversity networks performance analysis over Nakagami${-}m$ fading channels,"
\emph{IET Commun.}, vol. 5, no. 6, pp. 872${-}$878, June 2011. 

\bibitem{Sofotasios_9}
P. C. Sofotasios, and S. Freear, 
``The $\kappa-\mu$/gamma extreme composite distribution: A physical composite fading model,''
\emph{IEEE WCNC  `11},  Cancun, Mexico, Mar. 2011, pp. 1398$-$1401.

\bibitem{Add_4}
Y. Lee,  M. H. Tsai, and S. I. Sou, 
``Performance of decode-and-forward cooperative communications with multiple dual-hop Relays over Nakagami${-}m$ fading channels," \emph{IEEE Trans. Wireless Commun.}, vol. 8, no. 6,  pp. 2853${-}$2859, June 2009.

\bibitem{Sofotasios_11}
P. C. Sofotasios, and S. Freear, 
``The $\eta-\mu$/gamma composite fading model,''
\emph{IEEE ICWITS `10}, Honolulu, HI, USA, Aug. 2010, pp. 1$-$4.

\bibitem{D:Lee} 
Y. Lee and M-H. Tsai,
``Performance of decode-and-forward cooperative communications over Nakagami${-}m$ fading channels,"
 \emph{IEEE Trans. on Veh. Technol.}, vol. 58, no. 3, pp. 1218${-}$1228, March 2009.
 
 \bibitem {SWR} 
 R. Swaminathan, R. Roy, M.D.Selvaraj, ``Performance analysis of triple correlated selection combining for cooperative diversity systems," \emph{in Proc. IEEE international symposium on Wireless Commun.}, June 2013, pp. 5483${-}$5488.

 \bibitem{KYY}
 K. Yang, J. Yang, J. Wu, C. Xing,  Y. Zhou, ``Performance analysis of DF cooperative diversity system with OSTBC over spatially correlated                              
 Nakagami${-}m$ fading channels, \emph{IEEE Trans. Veh. Technol.}, vol. 63, no. 3,  pp. 1270${-}$1281,  Mar. 2014.

 \bibitem{RaymondH}
R. H. Y. Louie, Y. Li, H. A. Suraweera, and B. Vucetic, ``Performance analysis of beamforming in two hop AFrelay networks with
antenna correlation," \emph{ IEEE Trans. Wireless Commun.}, vol. 8, no.6, pp. 3132${-}$3141, June 2009.                  

 \bibitem{indoor}
 D. Liqin, W. Yang, Z. Jiliang, L. Limei, Li Xi, and Y. Dayong, 
``Investigation of spatial correlation for two-user cooperative communication in indoor office environment,"
\emph{in Proc. IEEE Int.  Conf. on  Commun. Technol.}, Nov. 2010, pp. 420${-}$423.

\bibitem{Theo_keyhole}
T. Q. Duong, H. A. Suraweera, T. A. Tsiftsis, H. J. Zepernick, and A. Nallanathan,
``OSTBC transmission in MIMO AF relay systems with keyhole and spatial correlation effects,"
\emph{in Proc. IEEE Int. Conf. on Commun.}, June 2011, pp. 1${-}$6.

\bibitem {YA} 
 Y. A. Chau, and K. Y-Ta. Huang, ``Performance of cooperative diversity on correlated dual-hop channels with an Amplify-and-Forward
relay over Rayleigh fading environments," \emph{ in Proc. IEEE Int. Conf.  on TENCON}, Nov. 2011, pp. 563${-}$567.

 \bibitem{HKR}
H. Katiyar and R. Bhattacharjee, ``Performance of two-hop regenerative relay network under correlated Nakagami${-}m$ fading at multi-antenna
relay,'' \emph{IEEE Commun.  Lett.}, vol. 13, no. 11, pp. 820${-}$ 822, Nov. 2009.

\bibitem{KY} K. Yang, J. Yang, J. Wu, and C. Xing, ``Performance analysis of cooperative DF relaying
over correlated Nakagami${-}m$ fading channels," \emph{in Proc. IEEE Int. Conf.  on Commun.}, June 2013, pp. 4973${-}$4977.

\bibitem{MPI}
Y. Chen, R. Shi, and M. Long, ``Performance analysis of amplify-and forward relaying with correlated links,'' \emph {IEEE Trans. Veh. Tech.}, vol. 62,
no. 5, pp. 2344 ${-}$ 2349, June 2013.

\bibitem{HAS}
H. A. Suraweera, D. S. Michalopoulos and G. K. Karagiannidis, ``Performance of distributed diversity systems with a single amplify-and-forward relay,'' \emph{IEEE Trans Veh. Technol.}, vol. 58, pp. 2603-2608, June, 2009. 
 
 \bibitem{B:Nakagami} 
M. Nakagami, 
``The $m$-distribution,${-}$ A general formula of intensity distribution of rapid fading,"
\emph{in Statistical Methods in Radio Wave Propagation}, W.G. Hoffman, Ed., Oxford, U.K.: Pergamon, 1960. 

 \bibitem{A:Simon}  
M. K. Simon  and M.-S. Alouni, 
\emph{Digital Communication over Fading Channels}, 2nd ed., New York: Wiley, 2005.
 
 \bibitem{SP}
``Selection procedures for the choice of radio transmission technologies of the UMTS," \emph{3GPP TR 30.03U,} ver. 3.2.0, 1998.


\bibitem{IDR}
I. D., R. Nagraj, G. G. Messier, and S. Magierowski, ``Performance Analysis of Relay-Assisted Mobile-to-Mobile Communication in
Double or Cascaded Rayleigh Fading,'' \emph{ in Proc. IEEE Pacific Rim Conference}, Aug. 2011, pp. 631${-}$636.

\bibitem{Z:Sadek} 
W. Su, A. K. Sadek, and K. J. Ray Liu,
 ``SER performance analysis and optimum power allocation for decode-and-forward cooperation protocol in wireless networks,"
 \emph{in Proc. IEEE Wir. Commun. and Netw. Conf.}, Mar. 2005, vol. 2,  pp. 984${-}$989.

\bibitem{Tables} 
I. S. Gradshteyn and I. M. Ryzhik, 
\emph{Table of integrals, series, and products}, 7th ed., New York: Academic, 2007.

 \bibitem{B:Tables} 
A. P. Prudnikov, Yu. A. Brychkov, and O. I. Marichev, 
\emph{Integrals and Series}, vol.1, Elementary Functions, 1st ed. New York: Academic, 2007.

\bibitem{W:Amine}
A. Mezghani and J. A. Nossek,
``Modeling and minimization of transceiver power consumption in wireless networks," 
\emph{in Proc. IEEE International ITG Workshop on Smart Antennas}, Feb. 2011. pp. 1${-}$8. 

\bibitem{GXX}
 S. Cui, A. J. Goldsmith, and A. Bahai, ``Energy-constrained modulation optimization,"
  \emph{IEEE Trans.Wireless Commu.}, vol. 4, no. 5, pp. 2349${-}$2360, Sep. 2005.

\bibitem{body} 
 S. Boyd and L. Vandenberghe, \emph{Convex Optimization}, Cambridge University Press, 1994.

 \end{document}